\documentclass[12pt]{article}

\usepackage[utf8]{inputenc}
\usepackage{amsmath}
\usepackage{amsfonts}
\usepackage{amsthm}
\usepackage{color}
\usepackage{graphicx}
\usepackage[normalem]{ulem}
\definecolor{luc}{RGB}{126, 51, 0}

        \theoremstyle{plain}
        
        \newtheorem{proposition}{Proposition}[section]

        \theoremstyle{remark}
        \newtheorem{remark}{\bf Remark}[section]
        \theoremstyle{remark}
        
        \theoremstyle{remark}

\newcommand{\demi}{\frac{1}{2}}

\newcommand{\R}{{\mathbb{R}}}

\newcommand{\dt}{\partial_t}

\newcommand{\cint}[1]{\langle #1 \rangle}

\newcommand{\eps}{\varepsilon}
\newcommand{\Kn}{{\rm Kn}}
\newcommand{\C}{{\cal C}}
\newcommand{\Id}{{\it Id}}

\newcommand{\correct}[1]{#1}
\newcommand{\correctp}[1]{#1}

\begin{document}

\begin{center}

{\bf A BGK model for high temperature rarefied gas flows}

\vspace{1cm}
C. Baranger$^1$, Y. Dauvois$^1$, G. Marois$^1$, J. Math\'e$^1$, J. Mathiaud$^{1,2}$, L. Mieussens$^2$

\bigskip
$^1$CEA-CESTA\\
15 avenue des Sabli\`eres - CS 60001\\
33116 Le Barp Cedex, France\\
{ \tt(celine.baranger@cea.fr, julien.mathiaud@u-bordeaux.fr)}\\

\bigskip
$^2$Univ. Bordeaux, CNRS, Bordeaux INP, IMB, UMR 5251, F-33400, Talence, France\\

{ \tt(Luc.Mieussens@math.u-bordeaux.fr)}

\end{center}

\paragraph{Abstract:} 
High temperature gases, for instance in hypersonic reentry flows, show
complex phenomena like excitation of rotational and vibrational energy
modes, and even chemical reactions. For flows in the continuous
regime, simulation codes use analytic or tabulated constitutive laws
for pressure and temperature. In this paper, we propose a BGK model
which is consistent with any arbitrary constitutive laws, and which is
designed to make high temperature gas flow simulations in the rarefied
regime. A Chapman-Enskog analysis gives the corresponding transport
coefficients. Our approach is illustrated by a numerical comparison
with a compressible Navier-Stokes solver with rotational and
vibrational non equilibrium. The BGK approach gives a deterministic
solver with a computational cost which is close to that of a simple
monoatomic gas.

\bigskip

\paragraph{Keywords:}
rarefied gas dynamics, polyatomic gas, BGK model, real gas effect, {second principle}

\tableofcontents

\section{Introduction}

For atmospheric reentry of space vehicles, it is important to estimate
the heat flux at the solid wall of the vehicle. In such hypersonic
flows, the temperature is very large, and the air flow, which is a
mixture of monoatomic and polyatomic gases, is modified by chemical
reactions. The characteristics of the mixture (viscosity and
specific heats) then depend on its temperature (see~\cite{nbk86}).

One way to take into account this variability is to use appropriate
constitutive laws for the air.  For instance, quantum mechanics allows
to derive a relation between internal energy and temperature that
accounts for activation of vibrational modes of the molecules (see~\cite{anderson}). When
the temperature is larger, chemical reactions occur, and if the flow
is in chemical equilibrium, empirically tabulated laws can be used to
compute all the thermodynamical quantities (pressure, entropy,
temperature, specific heats) in terms of density and internal energy, 
like the one
given in~\cite{anderson,hansen}. These laws give a closure of the
compressible Navier-Stokes equations, that are used for simulations
in the continuous regime, at moderate to low altitudes (see, for
example,~\cite{mykv88}).

In high altitude, the flow is in the rarefied or transitional regime,
and it is described by the
Boltzmann equation of Rarefied Gas Dynamics, also called the Wang-Chang-Uhlenbeck equation in
case of a reacting mixture. This equation is much too complex to be
solved by deterministic methods, and such flows are generally
simulated by the DSMC method~\cite{BS_2017}.  However, it is
attractive to derive simplified kinetic models that account for high
temperature effects, in order to obtain alternative and deterministic
solvers{: for such computations, it is necessary to capture dense zones with high temperatures and very rarefied zones with low temperatures}. Up to our knowledge, the first attempt to introduce non ideal
constitutive laws into a kinetic model has recently been
published in~\cite{rahimi16}. In this article, the authors define the
constant volume specific heat~$c_v$ as a third-order polynomial
function of the temperature of the gas, and derive a mesoscopic model
based on the moment approach. {A similar approach is
proposed in~\cite{KKA_2019} that gives a correct Prandtl number.}
 Simplified Boltzmann models for mixtures
of polyatomic gases have also been proposed in~\cite{andries2002,bisi2016,DESVILLETTES2005219},
however, high temperature effects are not addressed in these references. 

In this paper, {our goal is to construct models that are able to
  capture macroscopic effects as well as kinetic effects at a
  reasonable numerical cost, for an application to reentry flows.} We propose two ways to include high
temperature effects (vibrational modes, chemical reactions) in a
generalized BGK model.  

First, we show that vibrational modes can be
taken into account by using a temperature dependent number of degrees
of freedom. This can be used in a BGK model for polyatomic gases, but
we show that the choice of the variable used to describe the internal
energy of the molecules is fundamental here. This model allows us to
simulate a mixture of rarefied polyatomic gases (like the air) with
rotational and vibrational non equilibrium, {with a single
  distribution function for the mixture}. As a consequence, we are
able to simulate a polyatomic gas flow with a non-constant specific
heat.

Then we propose a more general BGK model that can be used to describe
a rarefied flow with both vibrational excitation and chemical
reactions, at chemical equilibrium, based on arbitrary constitutive
laws for pressure and temperature. Our BGK model is shown to be consistent
with the corresponding Navier-Stokes model in the small Knudsen number
limit.  Finally, the internal energy variable of our BGK model can be
eliminated by the standard reduced distribution
technique~\cite{HH_1968}: this gives a kinetic model for high
temperature polyatomic gases with a computational complexity which is
close to that of a simple monoatomic model. 

{Up to our knowledge, the model proposed in this work is the
  first Boltzmann model equation that allows for realistic
  equations of state and includes concentration effects in the thermal
  flux. We point out that this article is a first step towards a
  correct computation of the parietal heat flux: since we use a BGK
  model, it is clear that our model does not have a correct Prandtl
  number, as usual. This might be solved by using the ES-BGK
  approach~\cite{esbgk_poly,Mathiaud2016,Mathiaud2017} to capture the
  correct relaxation times for energy and fluxes~\cite{kustova}}.

The outline of our paper is the following. First we remind a standard
BGK model for polyatomic gases in
section~\ref{sec:polyatomicBGK}. Then, in section~\ref{sec:hightemp},
we explain how high temperature effects are taken into account to
define the internal number of degrees of freedom of molecules and
generalized constitutive laws.
A first BGK model is proposed to allow
for vibrational mode excitation with a temperature dependent number of
degrees of freedom in section~\ref{sec:deltaT}. This model is
extended to allow for arbitrary constitutive laws for pressure and
temperature in section~\ref{sec:BGK_EOS}, and this model is also
  analyzed by the Chapman-Enskog expansion. {Some features of
  our new model are illustrated by a few numerical simulations in
  section~\ref{sec:num}.}

\section{Polyatomic BGK model}
\label{sec:polyatomicBGK}
For
standard temperatures, a polyatomic perfect gas can be described by
the mass distribution function $F(t,x,v,\eps)$ that depends on time
$t$, position $x$, velocity $v$, and internal energy $\eps$. The
internal energy is described with a continuous variable, and takes into
account rotational modes. The corresponding number of degrees of
freedom for rotational modes is $\delta$ (see~\cite{BL1975}).

Corresponding macroscopic quantities mass density $\rho$, velocity $u$,
and specific internal energy $e$, are defined through the first 5 moments of $F$
with respect to $v$ and $\eps$: 
\begin{align}
  \rho(t,x) & = \cint{\cint{F}},\label{eq-rho}  \\
\rho u (t,x) & = \cint{\cint{v F}}, \label{eq-rhou} \\
\rho e (t,x) & = \cint{\cint{(\demi |v-u|^2 + \eps) F}} \label{eq-rhoe} ,
\end{align}
where
$\cint{\cint{\phi}} = \int_{\R^3}\int_0^{+\infty} \phi(v,\eps)\,
dvd\eps$
denotes the integral of any scalar or vector-valued function $\phi$
with respect to $v$ and $\eps$. The specific internal energy take into
account translational and rotational modes.  Other macroscopic
quantities can be derived from these definitions. The temperature is
$T$ is such that $e = \frac{3+\delta}{2}RT$, where $R$ is the gas
constant. The pressure is given by the perfect gas equation of state
(EOS) $p=\rho R T$. 

For a gas in thermodynamic equilibrium, the distribution function
reaches a Maxwellian state, defined by 
\begin{equation}  \label{eq-Maxwpoly}
M[F]  = \frac{\rho}{(2\pi R T)^{\frac{3}{2}}}
\exp\left( - \frac{|v-u|^2}{2RT}  \right) 
\Lambda(\delta)\left(\frac{\eps}{RT}\right)^{\frac{\delta}{2}-1} 
\frac{1}{RT}\exp\left( -\frac{\eps}{RT} \right),
\end{equation}
where $\rho$, $u$, and $T$ are defined above. The constant
$\Lambda(\delta)$ is a normalization factor defined by
$\Lambda(\delta) = 1/\Gamma(\frac{\delta}{2})$, so that $M[F]$
has the same 5 moments as $F$ (see above).

The simplest BGK model that can be derived from this description is
the following
\begin{equation}\label{eq-BGKpoly} 
\dt F + v\cdot \nabla_x F = \frac{1}{\tau}( M[F]-F),
\end{equation}
where $\tau$ is a relaxation time (see below).

The standard Chapman-Enskog expansion shows that this model is
consistent, \correct{with an error which is of second order with respect to the Knudsen number}, to
the following compressible Navier-Stokes equations
\begin{align*}
&   \dt \rho + \nabla\cdot \rho u = 0 \\
& \dt \rho u + \nabla \cdot (\rho u \otimes u) + \nabla p = -\nabla
\cdot \sigma \\ 
& \dt E + \nabla\cdot ((E+p)u) = -\nabla \cdot q - \nabla \cdot
(\sigma u), 
\end{align*}
where $\sigma = -\mu (\nabla u + (\nabla u)^T - \frac{2}{3+\delta}
\nabla \cdot u \, Id)$
is the shear stress tensor and $q = -\kappa \nabla T$ is the heat flux. The
transport coefficients $\mu$ and $\kappa$ are linked to the relaxation
time by the relations $\mu = \tau p$ and $\kappa = \mu c_p$, where the
specific heat at constant pressure is $c_p =
\frac{5+\delta}{2}R$. Actually, these relations define the correct
value that has to be given to the relaxation time $\tau$ of~(\ref{eq-BGKpoly}), which is
\begin{equation}  \label{eq-deftau}
\tau = \frac{\mu}{p},
\end{equation}
where the viscosity is given by a standard temperature dependent law
like $\mu(T) = \mu_{ref}(\frac{T}{T_{ref}})^{\omega}$
(see~\cite{bird}).  This implies that the Prandtl number
${\rm Pr} = \frac{\mu c_p}{\kappa}$ is equal to $1$. This incorrect
result (it should be close to $\frac{5}{7}$ for a diatomic gas, for
instance) is due to the fact that the BGK model contains only one
relaxation time. Instead it would be more relevant to include at least
three relaxation times in the model to allow for various different
time scales (viscous versus thermal diffusion time scale,
translational versus rotational energy relaxation rates). It is
possible to take these different time scales into account by using the
ESBGK polyatomic model (see~\cite{esbgk_poly}), or the Rykov model
(see~\cite{LYXZ_2014} and the references therein). See also multiple relaxation time BGK models developed for polyatomic gases in~\cite{ARS_2017,ARS_2018}. However, in this
work, the derivation of a model for high temperature gases is based on
this simple polyatomic BGK model (with a single relaxation time).

Note that this model is generally simplified by using the variable $I$ such
that the internal energy of a molecule is $\eps =
I^{\frac{2}{\delta}}$ (see~\cite{esbgk_poly}). Then the corresponding distribution ${\cal F}(t,x,v,I)$ is
defined such that ${\cal F}dxdvdI = Fdxdvd\eps$, which gives ${\cal F}
= I^{\frac{2}{\delta}-1}F$. The macroscopic quantities are defined by  
\begin{equation*}
  \rho(t,x)  = \cint{\cint{{\cal F}}}, \qquad 
\rho u (t,x)  = \cint{\cint{v {\cal F}}},\qquad
\rho e (t,x) = \cint{\cint{(\demi |v-u|^2 + I^{\frac{2}{\delta}}) {\cal F}}},
\end{equation*}
where now
$\cint{\cint{\phi}} = \int_{\R^3}\int_0^{+\infty} \phi(v,I)\,
dxdI$. The corresponding Maxwellian, which is simpler, is
\begin{equation}\label{eq-MaxwpolyI}
{\cal  M}[{\cal F}]  = \frac{\rho}{(2\pi R T)^{\frac{3}{2}}}
\exp\left( - \frac{|v-u|^2}{2RT}  \right) 
\frac{2}{\delta}\Lambda(\delta)\frac{1}{(RT)^{\frac{\delta}{2}}}
\exp\left( -\frac{I^{\frac{2}{\delta}}}{RT} \right).
\end{equation}
The corresponding BGK equation is
\begin{equation}\label{eq-BGKpolyI} 
\dt {\cal F} + v\cdot \nabla_x {\cal F} = \frac{1}{\tau}( {\cal M}[{\cal F}]-{\cal F}),
\end{equation}
which is equivalent to~\eqref{eq-BGKpoly}.

{Moreover, note that these models can be derived from~\cite{bourgat94}: in \correct{that} paper, the authors first give a Boltzmann collision operator for polyatomic gases deduced from the Borgnakke-Larsen model. In this model, the internal energy variable $\eps$ is described by a variable $I$ such that $I=\sqrt{\eps}$. By using the corresponding Maxwellian, it is easy to derive a single relaxation time BGK model. When this model is written with $\eps$, we exactly get model~\eqref{eq-BGKpoly}.}

{In the same paper~\cite{bourgat94}, the authors propose a second Boltzmann collision operator based on a model for a monoatomic gas in higher dimension, with an internal variable $w$ that lives in a $\delta$-dimension space, where $\delta$ is the number of internal degrees of freedom. The internal energy of the model is $\eps=|w|^2$. This model is written in polar coordinates $w=I\theta$, where $I$ is the norm of $w$ (and hence again the square root of $\eps$), and $\theta$ is the polar angle, and then it is reduced by integration with respect to $\theta$. The authors get a Boltzmann collision operator in which the distribution function is multiplied by a weight function $\phi(I)=I^{\delta -1}$. Again, a BGK model can be derived from this formulation, but it is different from models~\eqref{eq-BGKpoly} and~\eqref{eq-MaxwpolyI}. The resulting model has been extended by several authors to get BGK models for non polytropic gases (see section~\ref{sec:hightemp}). However, in the case of polytropic gases (i.e. constant $\delta$), this model can easily be shown to be equivalent to our model~\eqref{eq-BGKpoly}.
}

\section{High temperature gases}
\label{sec:hightemp}

When the temperature of the gas is larger, new phenomena appear
(vibration, chemical reactions, ionization). For instance, for
dioxygen, at $800$K, the molecules begin to vibrate, and chemical
reactions occur for much larger temperatures (for instance, dissociation of $O_2$
into $O$ starts at $2500$K).

The next sections explain how some of these effects (vibrations and
chemical reactions) can be taken into accounts in terms of
EOS and number of internal degrees of freedom.

   \subsection{Vibrations}

   Of course, the definition of the specific internal energy must
   account for vibrational energy. A possible way to do so is to
   increase the number of internal degrees of freedom $\delta$, that
   now accounts for rotational and vibrational modes. However, a
   result of quantum mechanics implies that this number of
   degrees of freedom is not an integer anymore, and that it is even
   not a constant (it is temperature dependent), see the examples
   below. Vibrating gases have other properties that make them quite
   different to what is described by the standard kinetic theory of
   monoatomic gases. For instance, the specific heat at constant
   pressure $c_p$ becomes temperature dependent. However, vibrating
   gases can still be considered as perfect gases, so that the perfect
   EOS $p=\rho R T$ still holds (in fact, such gases are called
   thermally perfect gases, see~\cite{anderson}).

   Now we give two examples of gases with vibrational excitation, and
   we explain how their number of internal degrees of
   freedom is defined.

   \subsubsection{Example 1: dioxygen}
   \label{subsubsec:example1}

At equilibrium, translational $e_{tr}$ and rotational
$e_{rot}$ specific energies can be defined by
\begin{equation*}
  e_{tr} = \frac{3}{2}RT \quad \text{ and } \quad e_{rot} = RT.
\end{equation*}
This shows that a molecule of dioxygen has $3$ degrees of freedom for
translation, and $2$ for rotation. 
By using quantum mechanics~\cite{anderson}, vibrational specific energy $e_{vib}$ is
found to be
\begin{equation*}
    e_{vib} = \frac{T^{vib}_{O_2}/T}{\exp(T^{vib}_{O_2}/T)-1} RT,
\end{equation*}
where $T^{vib}_{O_2}=2256$K is a reference temperature.

The number of ``internal'' degrees of freedom $\delta$, related to
rotation and vibration modes only, is defined such that the total
specific internal energy $e$ is
\begin{equation*}
    e = e_{tr}+e_{rot}+e_{vib} = \frac{3+\delta}{2}RT.
\end{equation*}
By combining this relation with the relations above, we find that
$\delta$ is actually temperature dependent, and defined by
\begin{equation*}
  \delta(T) = 2 + 2\frac{T^{vib}_{O_2}/T}{\exp(T^{vib}_{O_2}/T)-1}.
\end{equation*}

Accordingly, the specific heat at constant pressure $c_p$, which is
defined by $dh = c_pdT$, where the enthalpy is $h = e +
\frac{p}{\rho}$, can be computed as follows. Since $p=\rho R T$, we
find $h = \frac{5+\delta(T)}{2}RT$, and hence the enthalpy depends on
$T$ only, through a nonlinear relation. This means that $c_p= h'(T)$
is not a constant anymore, while we have $c_p = \frac{5+\delta}{2}R$
without vibrations. Finally, note that the relation that defines the
temperature $T$ through the internal specific energy $e =
\frac{3+\delta(T)}{2}RT $ now is nonlinear (it has to be
inverted numerically to find $T$).

   \subsubsection{Example 2: air}
   \label{subsubsec:example2}

   The air at moderately high temperatures ($T<2500$K) is a
   non-reacting mixture of nitrogen $N_2$ and dioxygen $O_2$. To
   simplify, assume that their mass concentrations are 
   $c_{N_2} = 75\%$ and $c_{O_2}=25\%$. These two species are perfect
   gases with their own gas constants $R_{N_2}$ and $R_{O_2}$. {The
     gas constant $R$ of the mixture can be defined by
     $R = c_{N_2}R_{N_2} + c_{O_2}R_{O_2}$ (see~\cite{anderson}).}

The specific internal energy is defined by $e=
c_{O_2}e_{O_2}+c_{N_2}e_{N_2}$. The energy of each species can be
computed like in our first example (see
section~\ref{subsubsec:example1}), and we find: 
\begin{equation*}
  e_{N_2} =  \frac{3+\delta_{N_2}(T)}{2}R_{N_2}T \qquad \text{and}
  \qquad 
 e_{O_2} =  \frac{3+\delta_{O_2}(T)}{2}R_{O_2}T,
\end{equation*}
where the number of internal degrees of freedom of each
species are
\begin{equation}\label{eq-deltaO2N2} 
\delta_{N_2}(T) = 2 + 2\frac{T^{vib}_{N_2}/T}{\exp(T^{vib}_{N_2}/T)-1}
 \qquad \text{and} \qquad 
\delta_{O_2}(T) = 2 + 2\frac{T^{vib}_{O_2}/T}{\exp(T^{vib}_{O_2}/T)-1}, 
\end{equation}
with $T^{vib}_{N_2}=3373$K and $T^{vib}_{O_2}=2256$K. Then the specific
internal energy of the mixture is
\begin{equation*}
\begin{split}
 e& =c_{O_2}\frac{3+\delta_{O_2}(T)}{2} R_{O_2}T+c_{N_2}\frac{3+\delta_{N_2}(T)}{2} R_{N_2}T \\
  &=\frac32RT+\frac12(c_{O_2}\delta_{O_2}(T)R_{O_2}+c_{N_2}\delta_{N_2}(T)R_{N_2})T
  \\
& = \frac{3+\delta(T)}{2}RT 
\end{split}
\end{equation*}
with the number of internal degrees of freedom given by
\begin{equation}\label{eq-deltaair} 
\begin{split}
   \delta(T) & =\frac{c_{O_2}\delta_{O_2}(T)  R_{O_2}+c_{N_2}\delta_{N_2}(T)R_{N_2}}{R} \\
    & =  2+\frac2R\left(c_{O_2}R_{O_2}\frac{T^{vib}_{O_2}/T}{\exp(T^{vib}_{O_2}/T)-1}+c_{N_2}R_{N_2}\frac{T^{vib}_{N_2}/T}{\exp(T^{vib}_{N_2}/T)-1}\right).
\end{split}
\end{equation}

We show in figure~\ref{fig:delta} the number of internal degrees of
freedom for each species and for the whole mixture. For all gases,
$\delta$ is equal to $2$ below $500$K, which means that only the rotational
modes are excited: each species is a diatomic gas with 2 degrees of freedom of
rotation, and the mixture behaves like a diatomic gas too. Then the number
of degrees of freedom increases with the temperature, and is greater than
\correct{$3$} for $T=3000$K. At this temperature, the number of degrees of
freedom for vibrations is \correct{around $1$}. Note that in addition to this graphical
analysis, it can be analytically proved
that all the $\delta$ computed here are increasing functions of $T$.

   \subsection{Chemical reactions}
   \label{subsec:chemical}

When chemical reactions have to be taken into account (for the air,
this starts at $2500$K), the perfect gas EOS still holds for each species,
but the EOS for the reacting mixture is less simple. To avoid the
numerical solving of the Navier-Stokes equations for all the species,
in the case of an equilibrium chemically reacting gas, it is convenient
to use instead a Navier-Stokes model for the mixture (considered as a
single species), for which tabulated EOS $p = p(\rho,e)$ and even
a tabulated temperature law $T=T(\rho,e)$ are used
(see~\cite{anderson}, chapter 11).

More precisely, it can be proved that for a mixture of thermally
perfect gases in chemical equilibrium, with a constant atomic nuclei composition, two state variables, like
$\rho$ and $e$, are sufficient to uniquely define the chemical
composition of the mixture. Let us precise what this means, with
notations that will be useful in the paper.
For each species of the mixture, numbered with index
  $i$: 
\begin{itemize}
  \item its concentration $c_i$ depends on $\rho$ and $e$ only: $c_i =
    c_i(\rho,e)$ ; 
  \item its pressure $p_i$ satisfies the usual perfect gas law: $p_i =
    \rho_i R_i T$, where $R_i$ is the gas constant of the species and $\rho_i = c_i(\rho,e) \rho$, so that $p_i =
    p_i(\rho,e)$ ; 
  \item its specific energy $e_i$ and enthalpy $h_i$ depend on $T$
    only: $e_i = e_i(T)$ and $h_i=h_i(T)$,  
where $e_i(T) = \frac{3+\delta_i(T)}{2} R_iT+e_i^{f,0}$, with
    $e_i^{f,0}$ is the energy of formation of the $i$th molecule and $\delta_i(T) $ is the number of
    activated internal degrees of freedom of the molecule that might
    depend on the temperature, see the previous sections ($\delta_i
    =0$ for monoatomic molecules).
\end{itemize}

For compressible Navier-Stokes equations for an equilibrium chemically
reacting mixture, these quantities are not necessary. Instead, it is
sufficient to define (with analytic formulas or tables): 
\begin{itemize}
  \item the total pressure $p = \sum_i p_i(\rho,e)$ so that $p =
    p(\rho,e) = \rho R(\rho,e) T$, with $R(\rho,e)= \sum_i
    c_i(\rho,e)R_i$ ;
      \item the temperature $T$, through the relation $e = \sum_i
    c_i(\rho,e)e_i(T)$, so that $T = T(\rho,e)$ ;
  \item the total specific enthalpy $h = \sum_i c_ih_i$, so that $h =
    h(\rho,e) = e + \frac{p(\rho,e)}{\rho}$.
\end{itemize}
We refer to~\cite{anderson} for details on this subject.

\section{BGK models for high temperature gases}

\subsection{A polyatomic BGK model for a variable number of degrees of freedom}
\label{sec:deltaT}

In this section, we propose an extension of the polyatomic BGK
model~(\ref{eq-BGKpoly}) to take into account temperature dependent
number of internal degrees of freedom, like in examples of
sections~\ref{subsubsec:example1} and~\ref{subsubsec:example2}.

This extension (already obtained in~\cite{hdr}) is quite obvious, since we just replace the constant
$\delta$ in~(\ref{eq-Maxwpoly}) by the temperature dependent
$\delta(T)$. For completeness, this model is given below:
\begin{equation}\label{eq-BGKpolydelta} 
\dt F + v\cdot \nabla_x F = \frac{1}{\tau}( M[F]-F),
\end{equation}
with 
\begin{equation}  \label{eq-Maxwpolydelta}
M[F]  = \frac{\rho}{(2\pi R T)^{\frac{3}{2}}}
\exp\left( - \frac{|v-u|^2}{2RT}  \right) 
\Lambda(\delta(T))\left(\frac{\eps}{RT}\right)^{\frac{\delta(T)}{2}-1} 
\frac{1}{RT}\exp\left( -\frac{\eps}{RT} \right).
\end{equation}
The macroscopic quantities are defined
by~(\ref{eq-rho})--(\ref{eq-rhoe}), while the temperature $T$ is
defined by
\begin{equation}\label{eq-eT} 
e = \frac{3+\delta(T)}{2}RT.
\end{equation}
Indeed, this implicit relation is invertible if, for instance, $\delta(T)$ is
an increasing function of $T$. This is true, at least for the
examples shown in section~\ref{subsubsec:example2}: it can easily be
shown that equations~(\ref{eq-deltaO2N2}) and~(\ref{eq-deltaair})
define increasing functions of $T$. 
Finally, the relaxation time $\tau$ is given by~(\ref{eq-deftau})
with $p=\rho R T$.

The same model has been proposed independently in~\cite{KKA_2019}, and
extended to an ES-BGK version to have correct transport coefficients.
\correct{Note that in~\cite{KKA_2019}, the temperature dependent
  number of degrees of freedom $\delta(T)$ is constructed through a
  given law for $c_v$ (the specific heat at constant volume), which is different from our approach.}

{
Note that model~(\ref{eq-MaxwpolyI})--(\ref{eq-BGKpolyI}) cannot be used here. Indeed, the change of variables $I = \varepsilon^{\frac{\delta(T)}{2}}$ now depends on time and space through $T$, and the corresponding model written with variable $I$ contains many more terms than~(\ref{eq-MaxwpolyI})--(\ref{eq-BGKpolyI}).} 

{Finally, we mention the alternate approach derived from the second polyatomic Boltzmann operator of~\cite{bourgat94}: in~\cite{ARS_2017, BRS_2018}, a weight function is used to fit any given non polytropic gas law. This requires to invert a Laplace transform, and is different from the approach presented here.
}


\subsection{A more general BGK model for arbitrary constitutive laws}
\label{sec:BGK_EOS}

In this section, we now want to extend the polyatomic BGK
model~(\ref{eq-BGKpoly}) so as to be consistent with arbitrary
constitutive laws $p = p(\rho,e)$ and $T=T(\rho,e)$ that can be used for
an equilibrium chemically reacting gas (see
section~\ref{subsec:chemical}). 

We define the gas constant of the mixture by
\begin{equation}  \label{eq-Rrhoe}
R(\rho,e) = \frac{p(\rho,e)}{\rho T(\rho,e)}
\end{equation} 
so that the EOS of perfect gases $p(\rho,e) = \rho R(\rho,e) T(\rho,e)$
  holds (note that a definition of $R$ from the concentrations and the gas
  constant of each species can also be used, see section~\ref{subsec:chemical}). We also
  note $\delta(\rho,e)$ the number of internal degrees of freedom
  defined such that $e =
  \frac{3+\delta(\rho,e)}{2}R(\rho,e)T(\rho,e)$.




Our BGK model is obtained by using the same approach as in
section~\ref{sec:deltaT}: we replace $R$ and $\delta$
in~\eqref{eq-Maxwpoly}--\eqref{eq-BGKpoly}
by their non constant values $R(\rho,e)$ and $\delta(\rho,e)$, so that
our model is
\begin{equation}\label{eq-BGK_EOS} 
\dt F + v\cdot \nabla_x F = \frac{1}{\tau(\rho,e)}( M[F]-F),
\end{equation}
with 
\begin{equation}  \label{eq-Maxwchimie}
M[F]  = \frac{\rho}{(2\pi \theta(\rho,e))^{\frac{3}{2}}}
\exp\left( - \frac{|v-u|^2}{2\theta(\rho,e)}  \right) 
\Lambda(\delta(\rho,e))\left(\frac{\eps}{\theta(\rho,e)}\right)^{\frac{\delta(\rho,e)}{2}-1} 
\frac{1}{\theta(\rho,e)}\exp\left( -\frac{\eps}{\theta(\rho,e)} \right),
\end{equation}
where the macroscopic quantities are defined by
\begin{equation} \label{eq-moments_chimie} 
 \rho(t,x)  = \cint{\cint{F}}, \qquad 
\rho u (t,x)  = \cint{\cint{v F}}, \qquad 
\rho e (t,x)  = \cint{\cint{(\demi |v-u|^2 + \eps) F}},  
\end{equation} 
the variable $\theta(\rho,e)$ is
\begin{equation}  \label{eq-theta}
\theta(\rho,e) = R(\rho,e)T(\rho,e),
\end{equation}
the number of internal degrees of freedom is
\begin{equation}  \label{eq-deltarhoe}
\delta(\rho,e) = \frac{2e}{R(\rho,e)T(\rho,e)} -3.
\end{equation}
and the relaxation time is
\begin{equation}  \label{eq-tauchimie}
\tau(\rho,e) = \frac{\mu(T(\rho,e))}{p(\rho,e)},
\end{equation}
while,  $p(\rho,e)$ and  $T(\rho,e)$ are given
by analytic formulas or numerical tables. 


\begin{remark} \label{rem:compat_modeles} 
  This model is more general than our previous
  model~\eqref{eq-BGKpolydelta}--\eqref{eq-Maxwpolydelta} defined to
  account for vibrations. In other
  words, model~\eqref{eq-BGKpolydelta}--\eqref{eq-Maxwpolydelta} can
  be written under the previous form. This is explained below. 

First, relation~\eqref{eq-eT} defines the temperature $T$ as a function of
$e$, which can be written $T=T(\rho,e)$. Then, the perfect gas
EOS $p=\rho R T(\rho,e)$ gives $p=p(\rho,e)$. 
Then, by definition of $T$, the number of internal degrees of freedom,
given by analytic laws~(\ref{eq-deltaO2N2}) or~(\ref{eq-deltaair}) for
instance, satisfies~\eqref{eq-eT}, and hence can be written $
\delta(T) = 2\frac{e}{RT} - 3 $,
which is exactly~(\ref{eq-deltarhoe}). Moreover, the
relaxation time $\tau$ given by~(\ref{eq-deftau}) is compatible
with definition~(\ref{eq-tauchimie}). Finally, the Maxwellian defined
by~(\ref{eq-Maxwpolydelta}) is clearly compatible with
definition~(\ref{eq-Maxwchimie}). 

Consequently,  the analysis given in the next sections will be made
with this more general
model~(\ref{eq-BGK_EOS})--(\ref{eq-tauchimie}) only.
\end{remark}

   \subsection{Compressible  Navier-Stokes asymptotics}
   \label{subsec:CE}

In this section, we prove the following formal result. 
\begin{proposition} \label{prop:CE}
  The moments of $F$, solution of the BGK
  model~\eqref{eq-BGK_EOS}--\eqref{eq-tauchimie}, satisfy the
  following Navier-Stokes equations, up to $O(\Kn^2)$:
\begin{equation} \label{eq:NS_chimie}
\begin{split}
& \dt \rho + \nabla \cdot \rho u = 0, \\
& \dt \rho u + \nabla\cdot  (\rho u\otimes u) + \nabla p = -\nabla
\cdot \sigma , \\
& \dt E + \nabla \cdot (E+p)u = -\nabla\cdot  q -
\nabla\cdot(\sigma u) ,
\end{split}
\end{equation}
where $\Kn$ is the Knudsen number (defined below), $E$ is the total
energy density defined by $E  =
\demi \rho |u|^2 + \rho e$, and $\sigma$ and $q$ are the shear stress
tensor and heat flux vector defined by
\begin{equation} \label{eq-sigmaq} 
\begin{split}
& \sigma = -  \mu \left(\nabla u + (\nabla u)^T
      -\C \nabla \cdot u \, \Id\right) ,\\
& q =   - \mu\nabla h,
\end{split}
\end{equation}
with $h= e + \frac{p(\rho,e)}{\rho}$ is the
enthalpy, and $\C =
\frac{\rho^2}{p(\rho,e)} \partial_{\rho}(\frac{p(\rho,e)}{\rho})
+ \partial_e(\frac{p(\rho,e)}{\rho})$. 
\end{proposition}

Note that this result is consistent with the Navier-Stokes equations
obtained for non reacting gases. For instance, in case of
a thermally perfect gas, i.e when the enthalpy depends only on the
temperature (see~\cite{anderson}), we find that the
heat flux is $q =   - \kappa\nabla T(\rho,e)$, where the heat transfer
coefficient is $\kappa = \mu c_p$, with the heat capacity at constant
pressure is $c_p = h'(T)$. In such case, the Prandtl number, defined
by ${\rm Pr} = \frac{\mu c_p}{\kappa}$, is 1, like in usual BGK
models.

Moreover, this result gives a volume viscosity (also called second
coefficient of viscosity or bulk viscosity) which is $\omega = \mu(\frac{2}{3}
- \C)$. In the case of a gas with a constant $\delta$, like in a non
vibrating gas, this gives $\C = \frac{2}{3 + \delta}$, and hence
$\omega = \frac{2\delta}{3(\delta +3)}\mu$. For a monoatomic gas,
$\delta = 0$, and we find the usual result $\omega =0$.

This result is proved by using the standard Chapman-Enskog
expansion. The main steps of this proof are given in
sections~\ref{subsubsec:ndf} to~\ref{subsubsec:ns}, while some
technical details are given in appendix~\ref{app:gaussian}.

\subsubsection{Comments on this model}
\label{subsubsec:comments}

{For reacting gases, our model is consistent with the fact that the
  energy flux accounts for energy transfer by diffusion of chemical
  species. Indeed, if we assume that our constitutive laws satisfy the
  relations given in section~\ref{subsec:chemical}, then 
the enthalpy
  $h$ that appears in the heat flux
  in~\eqref{eq-sigmaq} is also $h=\sum_i c_i h_i$. Since $h_i$ is a function
  of $T$ only, we have
\begin{align*}
q & = - \mu\nabla h  = - \mu \sum_i    c_i\nabla h_i 
                       -\mu \sum_i h_i\nabla c_i \\
&  = - \mu c_p\nabla T -\mu \sum_i  h_i\nabla c_i,
\end{align*}
where $c_p = \sum_ic_i h'_i(T)$.}

{Some standard compressible Navier-Stokes solvers for reacting gases in
chemical equilibrium use the following heat flux
\begin{equation*}
  q =  - \kappa \nabla T  + \sum_i  \rho c_i U_i h_i,
\end{equation*}
in which the diffusion velocity $U_i$ can be modeled by the Fick law $\rho c_i
U_i=-\rho D_i \nabla c_i$ (see~\cite{anderson}), where $D_i$ is the
diffusion coefficient of the $ith$ species. }

{Our heat flux can indeed be written under the same form, with
$\kappa = \mu c_p$, and $D_i = \mu / \rho$. Consequently, the Prandtl
number ${\rm Pr}= \mu c_p/\kappa$ and Schmidt numbers
$S_i = \mu / \rho D_i$ are all equal to 1, which is the consequence of
our single time relaxation in our model. Usually Schmidt and Prandtl
numbers are very close, and hence recovering a correct Prandtl number
with an ESBGK-like approach should also give more correct the Schmidt
numbers.}
  
\correct{However, note that the compressible Navier-Stokes model with heat flux
given by the formula above leads to a violation of the second law
of thermodynamics. Indeed, classical theory of non equilibrium
thermodynamics states that the heat flux can only be given by a
temperature gradient, so that the physical entropy production due to
the heat flux ($-\frac{q}{T^2}\cdot \nabla T$) is non negative
(see~\cite{Struchtrup_moments}). Here, the heat flux depends on
$c_i$, and hence on $\rho$. This implies that $q$ contains a
$\nabla\rho$ term that will induce a $\nabla\rho\cdot \nabla T$ term,
which has an undefined sign in the entropy production. This is clearly
in contradiction with the second law.}

\correct{This drawback} is consistent with
    the fact that we are not able to prove a H-theorem for our
    model. But we believe our model is still interesting, since in its
    hydrodynamic limit, it is consistent with compressible
    Navier-Stokes models that are used for atmospheric
    reentry. These models are also probably not compatible with the second
    principle too, due to some terms in the thermal flux that are usually
    neglected.

\correct{Another} drawback of this model is its physical inconsistency at
  equilibrium, as it can be seen with the following example. Consider
  a mixture of two inert gases and suppose they are at collisional
  equilibrium with each other: then the equilibrium distribution is
  the sum of two Maxwellian distributions with different molar masses
  so that it cannot be reduced to a single Maxwellian distribution. At
  the contrary, our model, which describes the mixture by a single
  distribution, will necessarily give a single Maxwellian at
  equilibrium.  In case of an air flow, the difference in molar mass
  of nitrogen and dioxygen is small (around 12\%), and our single
  Maxwellian should not be very different from the exact equilibrium.
  Of course, the same problem occurs with reacting gases at
  equilibrium, except if the concentration of the product of chemical
  reactions (like O, NO, etc.) is small enough.

\subsubsection{Non-dimensional form}
\label{subsubsec:ndf}

Now we start the proof of the result given in proposition~\ref{prop:CE}.
We choose a characteristic length $x_*$, mass density $\rho_*$, and
energy $e_*$. This induces characteristic values for 
pressure $p_* =\rho_*e_*$, temperature $T_* = T(\rho_*,e_*)$, 
molecular and bulk
velocities $v_* = u_* = \sqrt{e_*}$, time $t_* = x_*/v_*$, internal energy $\eps_* =
e_*$, viscosity $\mu_* = \mu(T_*)$, relaxation time $\tau_* = \mu_*/p_*$, and distribution
$F_* = \rho_*/e_*^{5/2}$.

 By using the non-dimensional variables $w' = w/w_*$ (where $w$ stands for any
variables of the problem),
model~\eqref{eq-BGK_EOS}--\eqref{eq-tauchimie} can be written
\begin{equation}\label{eq-BGKpolydelta_adim} 
\partial_{t'} F' + v'\cdot \nabla_{x'} F' = \frac{1}{\Kn\, \tau'(\rho',e')}( M'[F']-F'),
\end{equation}
with 
\begin{equation}  \label{eq-Maxwchimie_adim}
M'[F']  = \frac{\rho'}{(2\pi \theta')^{\frac{3}{2}}}
\exp\left( - \frac{|v'-u'|^2}{2\theta'}  \right) 
\Lambda(\delta')\left(\frac{\eps'}{\theta'}\right)^{\frac{\delta'}{2}-1} 
\frac{1}{\theta'}\exp\left( -\frac{\eps'}{\theta'} \right),
\end{equation}
where the macroscopic quantities are defined by
\begin{equation} \label{eq-moments_chimie_adim} 
 \rho'(t',x')  = \cint{\cint{F'}}, \qquad 
\rho' u' (t',x')  = \cint{\cint{v' F'}}, \qquad 
\rho' e' (t',x')  = \cint{\cint{(\demi |v'-u'|^2 + \eps') F'}},  
\end{equation} 
the variable $\theta'$ is
\begin{equation}  \label{eq-theta_adim}
\theta' = R'T' 
\end{equation}
the number of internal degrees of freedom is
\begin{equation}  \label{eq-deltarhoe_adim}
\delta' = \delta = \frac{2e'}{R'T'} -3, 
\end{equation}
and the relaxation time is
\begin{equation}  \label{eq-tauchimie_adim}
\tau' = \frac{\mu'}{p'},
\end{equation}
while $p' = p(\rho_*\rho',e_*e')/\rho_*e_*$,
$T' = T(\rho_*\rho',e_*e')/T_*$, $R' = p'/\rho'T'$, and $\mu' =
\mu(T(\rho,e))/\mu_*$. Finally, the Knudsen number $\Kn$ that appears in~(\ref{eq-BGKpolydelta_adim})
is defined by
\begin{equation}  \label{eq-Kn}
\Kn = \frac{\tau_*}{t_*} = \frac{\lambda_*}{x_*},
\end{equation}
where $\lambda_* = \tau_*v_*$ can be viewed as the mean free path.

Note that, to simplify the notations, the dependence of ($p'$,
$\theta'$, $R'$, $\delta'$, $\tau'$, $p'$, $T'$) on $\rho'$ and $e'$ is
not made explicit any more in the previous expressions. Moreover, in
the sequel, the primes will be removed too.

\subsubsection{Conservation laws}
\label{subsubsec:conslaw}

The conservation laws induced by the non-dimensional BGK
model~(\ref{eq-BGKpolydelta_adim}) are obtained by
multiplying~(\ref{eq-BGKpolydelta_adim}) by $1$, $v$, and $(\demi
|v|^2+ \eps)$, and then by integrating it with respect to $v$ and
$\eps$. By using the Gaussian integrals given in
appendix~\ref{app:gaussian}, we get
\begin{equation}\label{eq-cons} 
\begin{split}
& \dt \rho + \nabla\cdot \rho u = 0, \\
& \dt \rho u + \nabla \cdot (\rho u\otimes u + \Sigma(F)) = 0, \\
& \dt E + \nabla \cdot (Eu + \Sigma(F) u + q(F)\cdot u) = 0,
\end{split}
\end{equation}
where the stress tensor $\Sigma(F)$ and the heat flux vector $q(F)$
are defined by 
\begin{align}
  & \Sigma(F) = \cint{\cint{(v-u)\otimes (v-u)F}}, \label{eq-SigmaF}  \\
  & q(F) = \cint{\cint{(\demi |v-u|^2+\eps)(v-u)F}} \label{eq-qF} .
\end{align}

\subsubsection{Euler equations}
\label{subsubsec:euler}

The Euler equations of compressible gas dynamics can be obtained as
follows. Equation~(\ref{eq-BGKpolydelta_adim}) implies the first order
expansion $F =
M[F]+O(\Kn)$, and hence $\Sigma(F) = \Sigma(M[F]) + O(\Kn)$ and $q(F)
= q(M[F])+O(\Kn)$. Using Gaussian integrals given in
appendix~\ref{app:gaussian} gives
\begin{equation*}
  \Sigma(M[F]) = p \Id, \quad \text{ and } \quad
q(M[F]) = 0.
\end{equation*}
Consequently, the conservation laws~(\ref{eq-cons}) yields
\begin{equation*}
\begin{split}
& \dt \rho + \nabla\cdot \rho u = 0, \\
& \dt \rho u + \nabla \cdot (\rho u\otimes u ) + \nabla p = O(\Kn), \\
& \dt E + \nabla \cdot( (E+p)u) = O(\Kn),
\end{split}
\end{equation*}
that are the Euler equations of compressible gas dynamics, up to
$O(\Kn)$ terms, with the given EOS $p=p(\rho,e)$.

For the following, it is useful to rewrite these equations as
evolution equations for non-conservatives variables $\rho$, $u$, and
$\theta$. After some algebra, we get
\begin{equation}\label{eq-cons-theta} 
 \begin{split}
& \dt \rho + u \cdot\nabla \rho  = - \rho \nabla \cdot u, \\
& \dt u + (u \cdot \nabla) u  = -\frac{1}{\rho}\nabla p + O(\Kn), \\
& \dt \theta + u\cdot \nabla \theta = -\theta \C \nabla \cdot u + O(\Kn),
\end{split}
\end{equation}
where $\C$ is given by 
\begin{equation}  \label{eq-C}
\C = \frac{\rho}{\theta}\partial_{\rho} \theta + \partial_e \theta.
\end{equation}

\subsubsection{Navier-Stokes equation}
\label{subsubsec:ns}
Navier-Stokes equations are obtained by using the higher order expansion $F =
M[F]+ \Kn \, G $. Introducing this expansion
in~(\ref{eq-SigmaF}) and~(\ref{eq-qF}) gives
\begin{equation*}
  \Sigma(F) = p \Id + \Kn\, \Sigma(G), \quad \text{ and } \quad
q(F) = \Kn \, q(G).
\end{equation*}
Then we have to approximate $\Sigma(G)$ and $q(G)$ up to
$O(\Kn)$. This is done by using the expansion of $F$
and~(\ref{eq-BGKpolydelta_adim}) to get
\begin{equation*}
  G = -\tau(\dt M[F] + v \cdot \nabla_x M[F]) + O(\Kn).
\end{equation*}
This gives the following approximations
\begin{equation} \label{eq-Sigmaq_M} 
\begin{split}
&   \Sigma(G)  = -\tau \cint{\cint{(v-u)\otimes (v-u)(\dt M[F] + v \cdot \nabla_x M[F])}} +
O(\Kn) , \\
&   q(G)  = -\tau \cint{\cint{(\demi |v-u|^2+\eps)(v-u)(\dt M[F] + v \cdot \nabla_x M[F])}} + O(\Kn).
\end{split}
\end{equation}

Now, we have to make some long computations to reduce
these expressions to those given in~(\ref{eq-sigmaq}). We start with
the stress tensor $\Sigma(G)$. First, note that the Maxwellian $M[F]$ given
by~(\ref{eq-Maxwchimie_adim}) can be separated into $M[F] = M_{tr}[F]M_{int}[F]$, with
\begin{equation*}
 M_{tr}[F] = \frac{\rho}{(2\pi \theta)^{\frac{3}{2}}}
\exp\left( - \frac{|v-u|^2}{2\theta}  \right) 
, \quad \text{ and } 
\quad M_{int}[F] = \Lambda(\delta)\left(\frac{\eps}{\theta}\right)^{\frac{\delta}{2}-1} 
\frac{1}{\theta}\exp\left( -\frac{\eps}{\theta} \right).
\end{equation*}
It is useful to introduce the notations $\cint{\phi}_v =
\int_{\R^3}\phi(v)\, dv$ and $\cint{\psi}_{\eps} =
\int_0^{+\infty}\psi(\eps)\, d\eps$ for any velocity (resp. energy)
dependent function $\phi$ (resp. $\psi$).
Then it can easily be seen that
\begin{equation*}
\begin{split}
\cint{M_{int}[F]}_{\eps} = 1, \quad  \cint{\dt M_{int}[F] }_{\eps} =
0,
 \quad \cint{\nabla_x M_{int}[F] }_{\eps} = 0.
\end{split}
\end{equation*}
This implies that $\Sigma(G)$ reduces to 
\begin{equation} \label{eq-SigmaGMtr} 
   \Sigma(G)  = -\tau \cint{(v-u)\otimes (v-u)(\dt M_{tr}[F] + v \cdot \nabla_x M_{tr}[F])}_v +
O(\Kn).  
\end{equation}

Now it is standard to write $\dt M_{tr}[F]$ and $\nabla_x M_{tr}[F]$ as
functions of derivatives of $\rho$, $u$, and $\theta$, and then to use
Euler equations~(\ref{eq-cons-theta}) to write time derivatives as
functions of the space derivatives only. After some algebra, we get
\begin{equation*}
  \dt M_{tr}[F] + v \cdot \nabla_x M_{tr}[F] 
= \frac{\rho}{\theta^{\frac{3}{2}}}M_0(V)\left( A(V) \cdot
  \frac{\nabla_x \theta}{\sqrt{\theta}} 
+ B(V):\nabla_x u \right) + O(\Kn),
\end{equation*}
where
\begin{equation*}
\begin{split}
& V = \frac{v-u}{\sqrt{\theta}}, \qquad M_0(V) =
\frac{1}{(2\pi)^{\frac32}}\exp(-\frac{|V|^2}{2}) , \\
& A(V) = \left(\frac{|V|^2}{2}- \frac52\right)V, \qquad 
 B(V) = 
V\otimes V 
- \left( 
          \left( \frac{|V|^2}{2}- \frac32\right) \C + 1
  \right) \Id.
\end{split}
\end{equation*}
Then, we introduce the previous relations in~(\ref{eq-SigmaGMtr}) to
get
\begin{equation*}
  \Sigma_{ij}(G) = -\tau \rho \theta \cint{V_iV_j B(V)M_0}_V 
  \nabla_{x_j} u_i + O(\Kn),
\end{equation*}
where we have used the change of variables $v\mapsto V$ in the
integral (the term with $A(V)$ vanishes due to the parity of
$M_0$). Then standard Gaussian integrals (see appendix~\ref{app:gaussian}) give
\begin{equation*}
  \Sigma(G) = - 
\tau \rho \theta \left(\nabla u + (\nabla u)^T
      -\C \nabla \cdot u \, Id\right) + O(\Kn), 
\end{equation*}
which is the announced result, in a non-dimensional form.

For the heat flux, we use the same technique to reduce $q(G)$ as given
in~(\ref{eq-Sigmaq_M}) to 
\begin{equation*}
\begin{split}
q_i(G) & = 
-\tau \cint{(\demi |v-u|^2)(v_i-u_i)(\dt M_{tr}[F] + v_j \partial_{x_j}
  M_{tr}[F])}_v \cint{M_{int}[F]}_{\eps} \\
& \quad -\tau \cint{(v_i-u_i)(\dt M_{tr}[F] + v_j \partial_{x_j}
  M_{tr}[F])}_v \cint{\eps M_{int}[F]}_{\eps} \\
& \quad  - \tau \cint{(v_i-u_i)M_{tr}[F] v_j}_v \cint{\eps \partial_{x_j} M_{int}}_{\eps} \\
& = - \tau \cint{\demi |V|^2V_i  A_{j}M_0}_V\partial_{x_j}\theta
- \tau \cint{V_i  A_{j}M_0}_V\frac{\delta}{2}\partial_{x_j}\theta
-\tau \rho \theta\cint{V_iV_jM_0}_V \partial_{x_j} (\frac{\delta}{2}\theta),
\end{split}
\end{equation*}
where we have used the relation $\cint{\eps M_{int}[F]}_{\eps} = \frac{\delta}{2}\theta$.
Using again Gaussian integrals, we get
\begin{equation*}
q(G) = -\tau \rho \theta \nabla h + O(\Kn),
\end{equation*}
where $h=\frac{5+\delta}{2}\theta$ is indeed the enthalpy, since
definitions~(\ref{eq-theta_adim}) and~(\ref{eq-deltarhoe_adim}) imply
$h = e + p/\rho$.

To summarize, we have shown that the stress tensor and heat flux in
conservation laws~(\ref{eq-cons}) are
\begin{equation*}
\begin{split}
& \Sigma(F) = p \Id - \Kn \tau \rho \theta \left(\nabla u + (\nabla u)^T
      -\C \nabla \cdot u \, Id\right) + O(\Kn^2)\\
& q(F) = -\Kn  \tau \rho \theta \nabla h + O(\Kn^2).
\end{split}
\end{equation*}
Now, we can go back to the dimensional variables, and we find
\begin{equation*}
\begin{split}
& \Sigma(F) = p(\rho,e) \Id - \mu(T(\rho,e)) \left(\nabla u + (\nabla u)^T
      -\C \nabla \cdot u \, Id\right) + O(\Kn^2)\\
& q(F) = - \mu(T(\rho,e)) \nabla h(\rho,e) + O(\Kn^2),
\end{split}
\end{equation*}
where $h(\rho,e) = e + \frac{p(\rho,e)}{\rho}$ is the enthalpy, and $\C = \frac{\rho^2}{p(\rho,e)} \partial_{\rho}(p(\rho,e)/\rho)
+ \partial_e(p(\rho,e)/\rho)$. This concludes the proof of the result
given at the beginning of this section.

   \subsection{Entropy}
Here, we prove that our model~\eqref{eq-BGK_EOS} satisfies a
local entropy dissipation property. 
\begin{proposition}
 Let~$F$ be the solution of
 equation~\eqref{eq-BGK_EOS}--\eqref{eq-Maxwchimie}. Then the following inequality is satisfied:
\begin{equation}
 \cint{\cint{ (M[F]-F) \ln 
\left( \frac{2}{\delta} \eps^{1-\frac \delta 2 } F\right) }} \leq 0 \, .
\label{eq:second_principe}
\end{equation}
\end{proposition}

\begin{proof}
The left-hand side can be decomposed into
$$
 \cint{\cint{ (M[F]-F) \ln 
\left( \frac{2}{\delta} \eps^{1-\frac \delta 2 } F\right) }}
 = \cint{\cint{  (M[F]-F) \ln \left( \frac F{M[F]} \right) }}
 +\cint{\cint{  (M[F]-F) \ln 
\left( \frac{2}{\delta} \eps^{1-\frac \delta 2 } M[F]\right)  }} \, .\\$$
The first term in the right-hand side is non-positive because the logarithm is a non-decreasing function.
The second term vanishes since $M[F]$ and $F$ have the same first 5
moments:
\begin{align*}
\cint{\cint{ (M[F]-F) \ln 
\left( \frac{2}{\delta} \eps^{1-\frac \delta 2 } M[F]\right) }}
         &= \cint{\cint{ (M[F]-F) }} \ln ( c(\delta,\rho,\theta)) \\
& \quad - \frac 1\theta \cint{\cint{ (M[F]-F) \left( \frac{|v-u|^2}{2} +
              \eps \right) }}\\
         &=0,
\end{align*}
with~$c(\delta,\rho,\theta) = \frac 2 \delta \frac{\rho \Lambda(\delta)}{\sqrt{2\pi \theta}^3\theta^{\delta/2}}$, which does not depend on~$v$ nor on~$\eps$.
\end{proof}

\begin{remark}
This result does not imply the dissipation of a
global entropy, except, for example, if~$\delta$ is constant. 
In such a case, we can define the entropy $H(f) = \cint{\cint{h(F)}}$,
where $h(F) = F\ln
\left(\frac{2}{\delta}\eps^{1-\frac{\delta}{2}}F\right) - F$, and we have
\begin{equation*}
\begin{split}
\dt H(F) + \nabla \cdot \cint{\cint{v h(F)}} & = \cint{\cint{\dt h(F) +
    v \cdot \nabla_x h(F)}} \\
& = \cint{\cint{ (\dt F + v \cdot \nabla_x F)h'(F)}} \\
& = \frac{1}{\tau}\cint{\cint{(M[F]-F)h'(F)}} \leq 0,
\end{split}
\end{equation*}
from~(\ref{eq:second_principe}), since $h'(F) = \ln
    \left(\frac{2}{\delta}\eps^{1-\frac{\delta}{2}}F\right)$.

In the general case, $\delta$ depends on $t$ and $x$: therefore, the relation
$\dt h(F) = \ln
    \left(\frac{2}{\delta}\eps^{1-\frac{\delta}{2}}F\right) \dt F  $
    is not correct. Consequently, the local property~(\ref{eq:second_principe}) cannot be used. It is not
clear so far that our model satisfies a global dissipation
property. This problem was also noticed in~\cite{KKA_2019}.
\end{remark}

   \subsection{Reduced model}

For computational reasons, it is interesting to reduce the complexity
of model~(\ref{eq-BGK_EOS}) by using the usual reduced
distribution technique~\cite{HH_1968}. We define
reduced distributions 
$f(t,x,v) = \int_0^{+\infty}F(t,x,v,\eps)\, d\eps$ and 
$g(t,x,v) = \int_0^{+\infty}\eps F(t,x,v,\eps)\, d\eps$, and by
integration of~(\ref{eq-BGK_EOS}) w.r.t $\eps$, we can
easily obtain the following closed system of two BGK equations
\begin{equation}\label{eq-reduced}
\begin{split} 
&   \dt f + v\cdot \nabla_x f = \frac{1}{\tau}(M[f,g]-f), \\
&   \dt g + v\cdot \nabla_x g 
  = \frac{1}{\tau}( \frac{\delta}{2}
                            RT M[f,g]-g),
\end{split}
\end{equation}
where $M[f,g]$ is the translational part of $M[F]$ defined by
\begin{equation*}
   M[f,g] = \frac{\rho}{(2\pi RT)^{\frac{3}{2}}}
\exp\left( - \frac{|v-u|^2}{2 RT}  \right) ,
\end{equation*}
and the macroscopic quantities are defined by
\begin{equation} \label{eq-mtsfg} 
 \rho(t,x)  = \int_{\R^3} f \, dv, \qquad 
\rho u (t,x)  = \int_{\R^3} v f \, dv, \qquad 
\rho e (t,x)  = \int_{\R^3}( \demi |v-u|^2 f + g) \, dv,  
\end{equation} 
while $\delta$, $R$ and $\tau$ are still defined
by~(\ref{eq-deltarhoe}),~\eqref{eq-Rrhoe} and~(\ref{eq-tauchimie}). This reduced system
is equivalent to~(\ref{eq-BGK_EOS}), that is to say $F$ and $(f,g)$ have
the same moments. Moreover, the compressible Navier-Stokes asymptotics
obtained in section~\ref{subsec:CE} can also be derived from this
reduced system. 
Consequently, \correct{this system is the one we use for our numerical
tests presented in the following section}. 

\section{Numerical results}
\label{sec:num}

   \subsection{Moderate temperature flow: vibrating molecules}

A numerical scheme for model~(\ref{eq-reduced}) has been implemented
in the code of CEA-CESTA. {This code  is a deterministic code based
  on the works presented in~\cite{mieussens99,bchm2013} which solves the BGK
  equation in 3 dimensions of space and 3 dimensions in velocity with
  a second order finite volume scheme on structured meshes}. It is remarkable that the original code (for
non reacting gases, with no high temperature effects), presented
in~\cite{bchm2013}, can be very easily adapted to this new model. Only a
few modifications are necessary.

The goal of this section is to illustrate the capacity of our model
to account for some high temperature gas effects. We only consider the case
of a mixture of two vibrating, but non reacting, gases. A validation
of our model for reacting gases will be given in a further work.


Our test is a 2D hypersonic plane flow of air--considered as a mixture
of two vibrating gases, nitrogen and dioxygen--over a quarter of a
cylinder which is supposed to be isothermal (see
figure~\ref{amont}). Gas-solid wall interactions are modeled by the usual
diffuse reflection.
At the inlet,  the flow is defined by the data given in table~\ref{table:ci}.
\begin{table}[h]
  \centering

\begin{tabular}{|l|l|}
\hline
Mass concentration of $N_2$ ($c_{N_2}$) & $0.75$\\
\hline
Mass concentration of $O_2$ ($c_{O_2}$) & $0.25$\\
\hline
Mach number of the mixture & $10$\\
\hline
Velocity of the mixture & $2267m.s^{-1}$\\
\hline
Density of the mixture &  $3.059\times 10^{-4} kg.m^{-3}$\\
\hline
Pressure of the mixture& $11.22 Pa$\\
\hline
Temperature of the mixture& $127.6 K$\\
\hline
Temperature of the cylinder& $293K$\\
\hline
Radius of the cylinder& $0.1m$\\
\hline
\end{tabular}

  \caption{Hypersonic flow around a cylinder: initial data}
  \label{table:ci}
\end{table}

In this case, the vibrational energy is taken into account as
described in section~\ref{subsubsec:example2}. The corresponding
constitutive relations are obtained as explained in
remark~\ref{rem:compat_modeles}.

The flow conditions are such that molecules vibrate, but
no chemical reactions are active (temperatures go up to $3000$K
whereas chemical reactions occur at $5000$K at pressure $P=1$atm): our
thermodynamical approach is reasonable. Since the test case is dense
enough (the Knudsen number is around $0.01$) we can compare the new model
with a Navier-Stokes code (a 2D finite
volume code with structured meshes), in which are enforced the same viscosity and conductivity as
in compressible Navier-Stokes asymptotics derived from the BGK model
(see section~\ref{subsec:CE}). To validate the new model we have made
four different simulations:
\begin{itemize}
 \item a Navier-Stokes simulation without taking into account vibrations (called $NS1$),
\item a Navier-Stokes simulation  that takes into account vibrations (called $NS2$),
 \item a BGK simulation without taking into account vibrations (called $BGK1$),
\item a BGK simulation that takes into account vibrations (called $BGK2$).
\end{itemize}

The first comparison is between $NS1$ and $BGK1$, in order to show
that the two model are consistent in this dense regimes, when there
are no vibration energy. As it can be seen in figure \ref{NSBGK1}, the results agree very well.

The second comparison is between $NS2$ and $BGK2$ to show we still
have a good agreement when vibrations are taken into account. This is
what we observe in figure \ref{NSBGK2}. One can also observe that, due to
vibrations, the temperature decreased from $2682K$ to $2358K$ for
Navier-Stokes and from $2695K$ to $2365K$ for BGK.

The last comparison is to show the influence of vibrational energy on
the results. We compare $BGK1$ and $BGK2$, and we observe that the shock is not
 at the same position. Since there is a transfer of energy from
 translational and rotational modes to vibrational modes, the maximum
 temperature is lower and the shock is slightly close to the cylinder
 with BGK2 (see figure~\ref{bgkbgk2}). We clearly see this difference
 with the temperature profile along the stagnation line, see figure~\ref{stagnation}.

{To conclude this section, it can be said that when Navier-Stokes
  and BGK are set with the same viscosity and Prandtl number, results
  agree very well: but of  course for more realistic test cases when
  the Prandtl number is not equal to one, there will be a discrepancy
  in the results that might be  corrected with an ES-BGK extension of
  our model. This will be presented in a further work.}

\correctp{
   \subsection{High temperature flow: reacting gas}
   \label{subsec:result_chemical}
In this section, we illustrate the ability of our model to account for
chemical reactions in a high temperature flow. In order to simplify
the analysis of our results, we consider here a single species flow of
dioxygen. The geometry of the test case is the same as in the previous
section, and the parameters of the upstream flow are the followings: the Mach
number is $12$, the density is $10^{-3}$ kg.m$^{-3}$, so that
the flow is in the near continuous regime (Kn$=4.29\times 10^{-4}$),
the pressure is $33.15$ Pa, the temperature is $127.6$ K, and the
temperature of the cylinder is still $283$ K.

In this case, the chemical reactions are taken into account with
pressure and temperature laws as given by Hansen~\cite{hansen}, both
in our Navier-Stokes and BGK solvers. We obtain the comparison shown
in figure~\ref{NSchimie_BGKchimie} for the temperature field. The
results given by both codes are very close. A closer look at the
temperature profile along the stagnation line is also shown in
figure~\ref{T_profile}: this profile shows that BGK results are
excellent.

We are also able to obtain the concentration $c_O$ of monoatomic oxygen
(see section~\ref{subsec:chemical}), and this concentration is
plotted in figure~\ref{CO_NSchimie_BGKchimie}. Again both codes are in
very good agreement, and these results show that there is dissociation
of $O_2$ molecules in the largest temperature zones, since the concentration
rises up to 12\% there. 

Finally, the importance of chemical reactions (dissociation) in this
test case can be seen as follows. In figure~\ref{BGKvibra_BGKchimie},
we compare the previous BGK results to a simulation made when
vibrations are taken into account but the chemical reactions are
not. This figure clearly shows that the non reacting BGK results are
incorrect: the location of the shock is wrong, and the temperature is
too high.

}
\section{Conclusion}
\label{sec:conclusion}

In this paper, we have proposed several generalized BGK models to
account for high temperature effects (vibrations and chemical
reactions). The first model is able to account for the fact that, for
polyatomic gases, some internal degrees of freedom are partially
excited with a level of excitation that depends on the temperature.
In other words, we have derived a model for a polyatomic gas with a
non-constant specific heat~$c_p = c_p(T)$.

This model has been extended to take into account general
constitutive laws for pressure and temperature, like in equilibrium
chemically reacting gases in high temperature flows.  By using a
Chapman-Enskog analysis, we have derived compressible Navier-Stokes
equations from this model that are consistent with these constitutive
laws. This consistency has been illustrated on preliminary numerical
tests, in which the importance to take vibration modes into account is
clearly seen.

We point out that this new model can be reduced to a BGK system in which the
molecular velocity is the only kinetic variable. This makes it
possible to simulate a high temperature polyatomic gas for the
cost of a simple monoatomic rarefied gas flow simulation.

The model for chemically reacting gases \correct{has been
  tested with a single species flow that shows its ability to account
  for dissociation, at least in the near continuum regime. Our model
has still to be validated with comparisons to a full Boltzmann (DSMC)
solver in the rarefied regime. It should also be }extended to allow for various
different time scales (viscous versus thermal diffusion time scale,
translational versus rotational energy relaxation rates). This might
be possible with the same approach as the one used to derive the ES-BGK
model for polyatomic gases (see~\cite{esbgk_poly}).

\appendix

\section{Gaussian integrals}
\label{app:gaussian}

We remind the definition of the absolute Maxwellian $M_0(V) =
\frac{1}{(2\pi)^{\frac32}}\exp(-\frac{|V|^2}{2})$. It is standard to
derive the following integral relations (see~\cite{chapmancowling},
for instance), written with the Einstein notation:
\begin{align*}
&   \cint{M_0}_V = 1, \\
&   \cint{V_iV_jM_0}_V = \delta_{ij}, \qquad \cint{V_i^2M_0}_V = 1,
  \qquad \cint{|V|^2M_0}_V = 3, \\
& \cint{V_i^2V_j^2M_0}_V = 1 + 2\, \delta_{ij},  \qquad \cint{V_iV_jV_kV_lM_0}_V = \delta_{ij}\delta_{kl}  +
  \delta_{ik}\delta_{jl}  + \delta_{il}\delta_{jk}  \\
& \cint{V_iV_j|V|^2M_0}_V = 5 \,\delta_{ij},  \qquad \cint{|V|^4M_0}_V
  = 15, \\
& \cint{V_iV_j|V|^4M_0}_V = 35 \,\delta_{ij},  \qquad \cint{|V|^6M_0}_V = 105,
\end{align*}
while all the integrals of odd power of $V$ are zero.

From the previous Gaussian integrals, it can be shown that for any
$3\times 3$ matrix $C$, we have
\begin{equation*}
\cint{V_iV_jC_{kl}V_kV_lM_0}_V = C_{ij} + C_{ji} + C_{ii}\delta_{ij}.
\end{equation*}

 \bibliographystyle{plain}

 \bibliography{biblio}


\begin{figure}[p]
\begin{center}
  \includegraphics[width=0.7\textwidth]{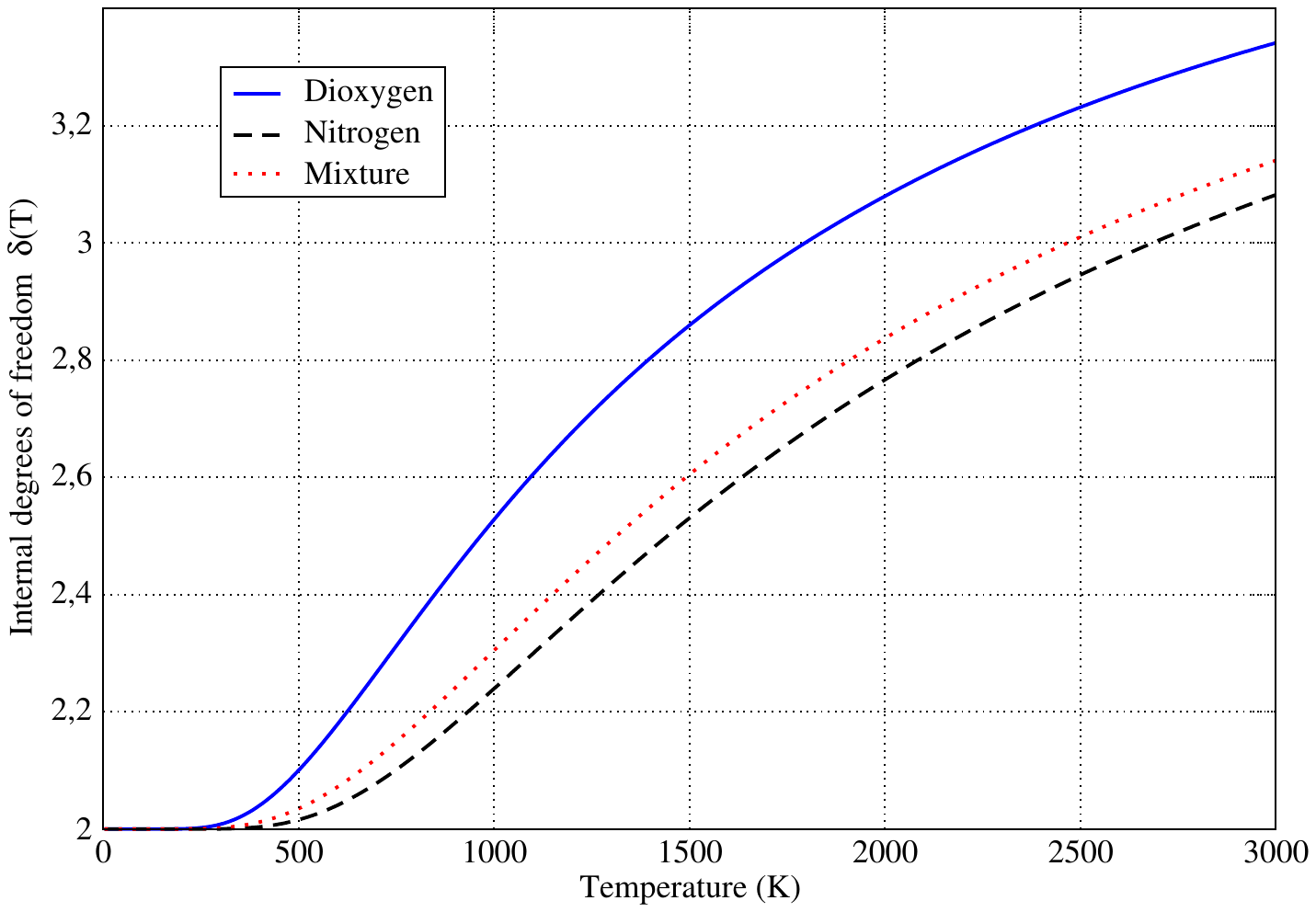}
\end{center}
\caption{Internal degrees of freedom as a function of the temperature}
\label{fig:delta}
\end{figure}

\clearpage

\begin{figure}[p]
\begin{center}
 \includegraphics[width=0.7\textwidth]{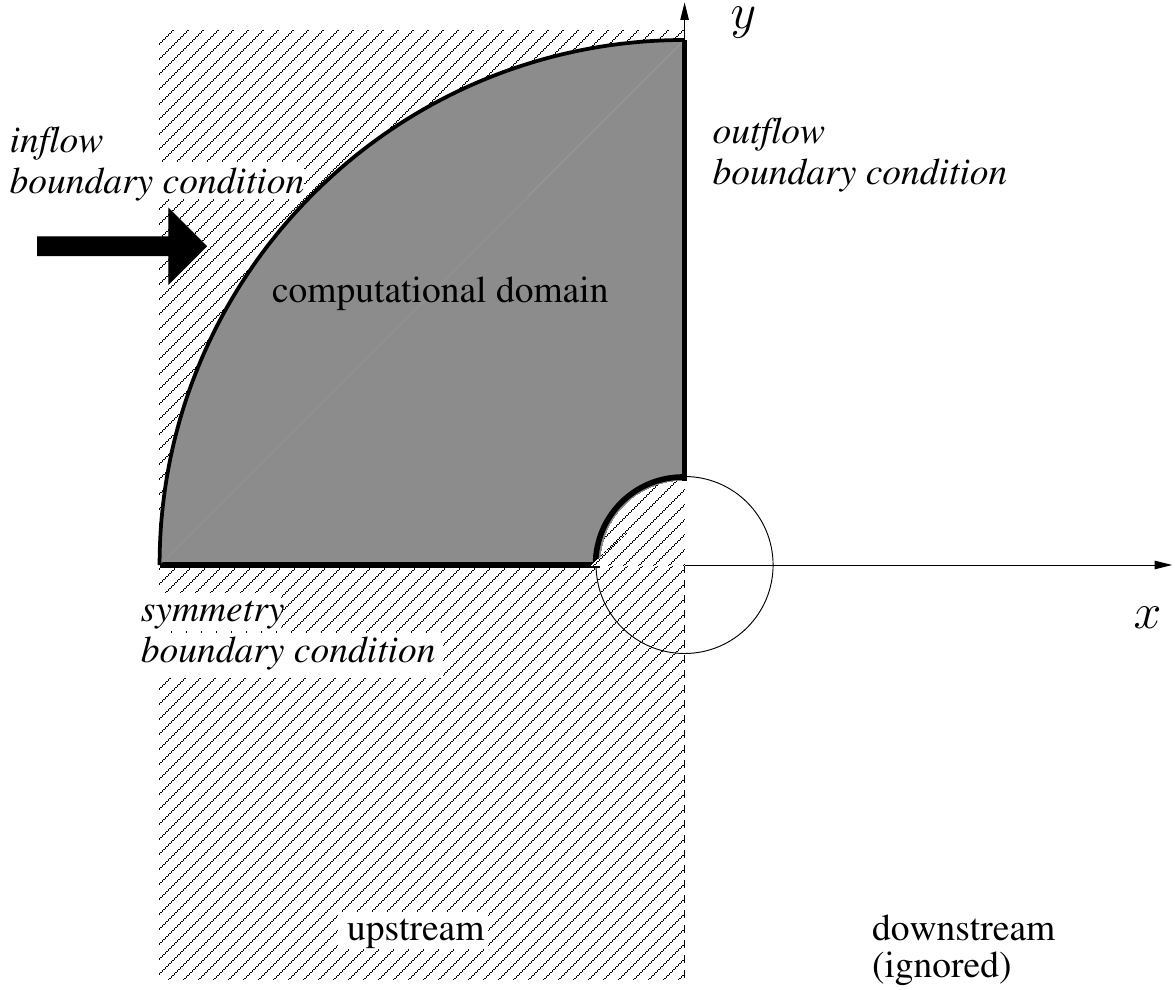}
\end{center}
\caption{
 \correct{Plane flow around a cylinder: geometry and computational domain. By
symmetry with respect to the horizontal axis,  the computational domain
is defined for the upper part only. The downstream flow is not simulated.}}
\label{amont}
\end{figure}

\clearpage

\begin{figure}[p]
 \centering
 \includegraphics[width=0.9\textwidth]{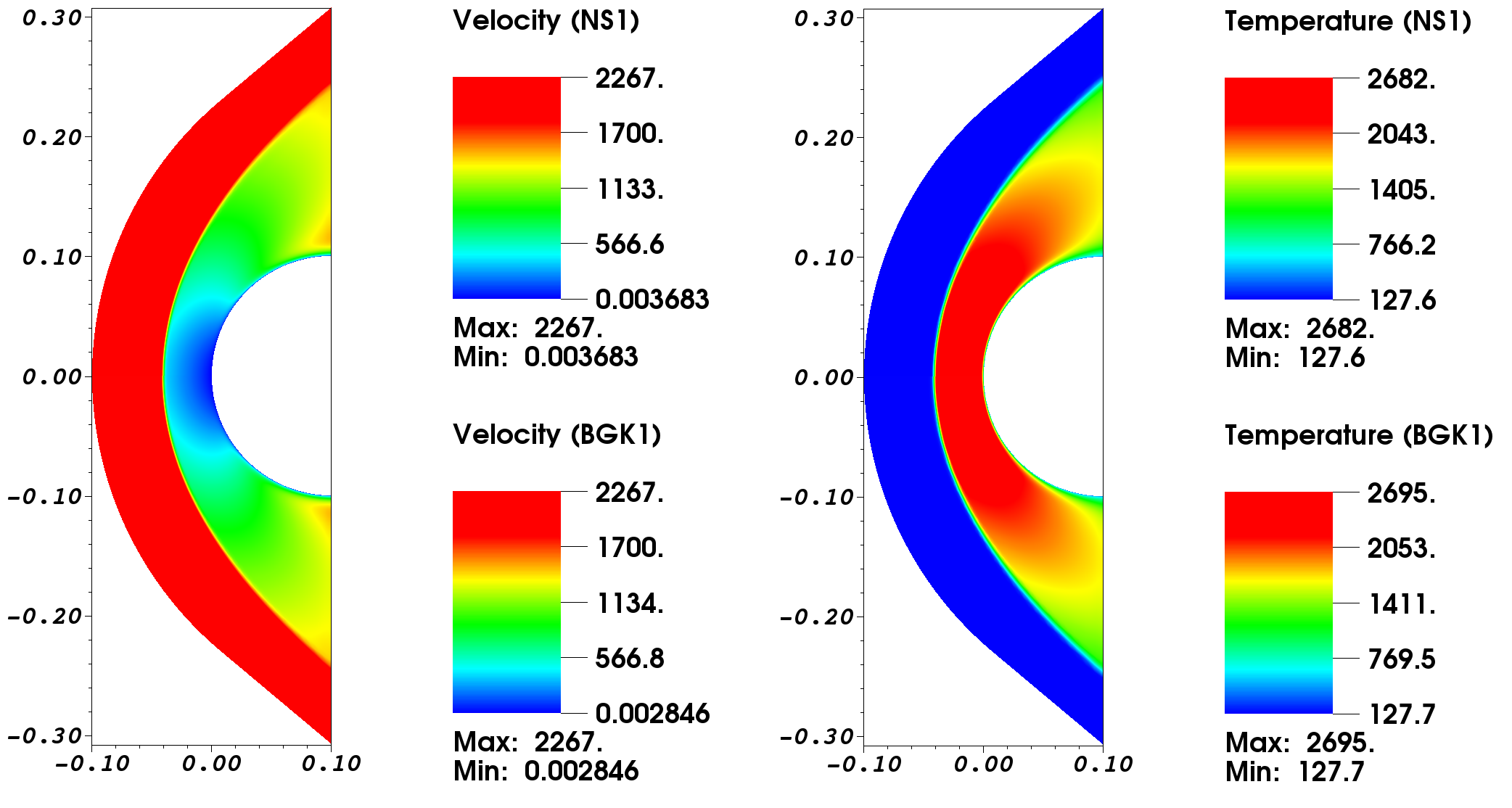}
 \caption{Non vibrating air: velocity and temperature fields (Top: NS1, bottom: BGK1)}
\label{NSBGK1}
\end{figure}

\clearpage

\begin{figure}[p]
 \centering
 \includegraphics[width=0.9\textwidth]{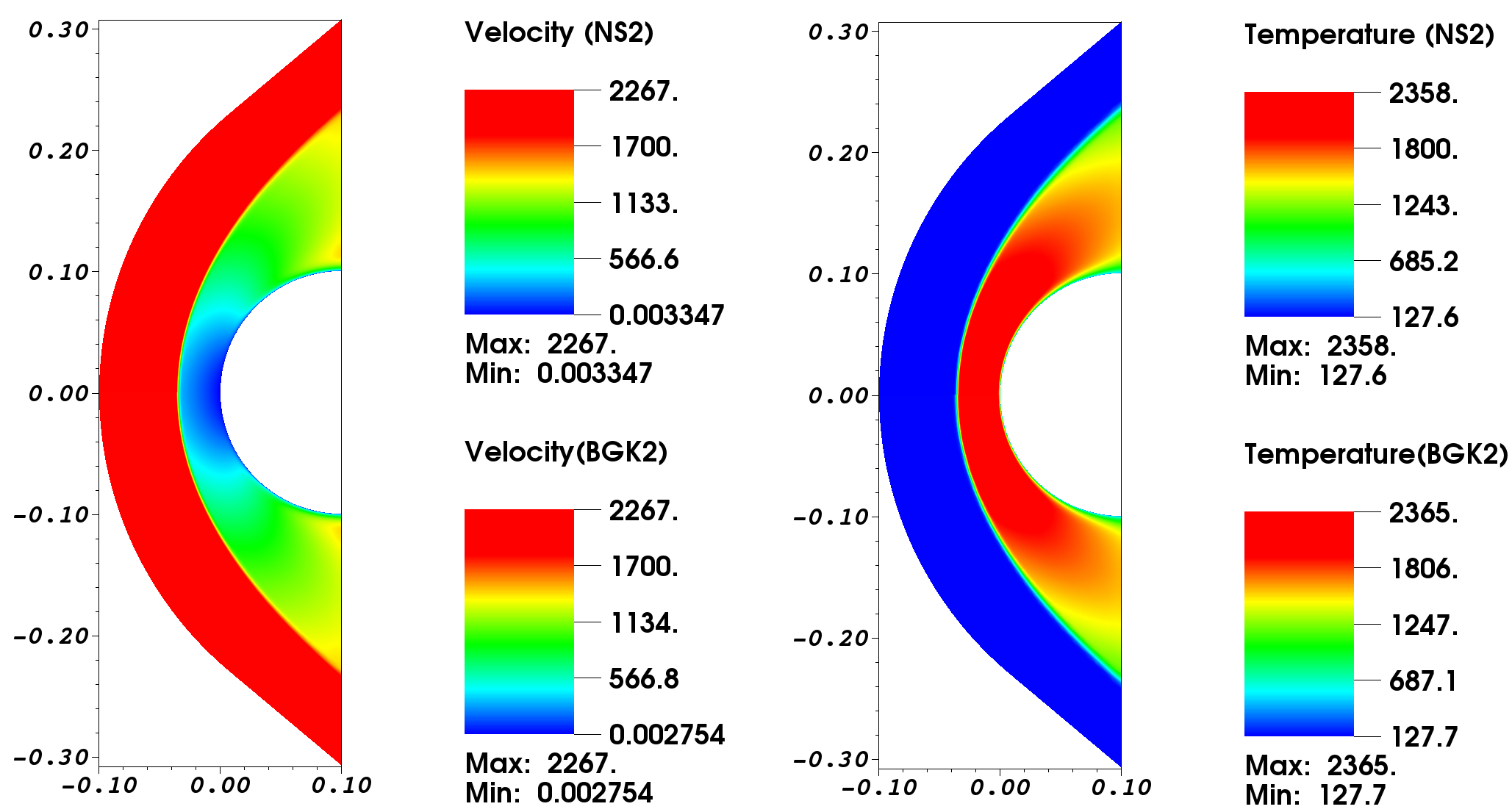}
 \caption{Vibrating air: velocity and temperature fields (Top: NS2, bottom: BGK2)}
\label{NSBGK2}
\end{figure}

\clearpage

\begin{figure}[p]
 \centering
 \includegraphics[width=0.9\textwidth]{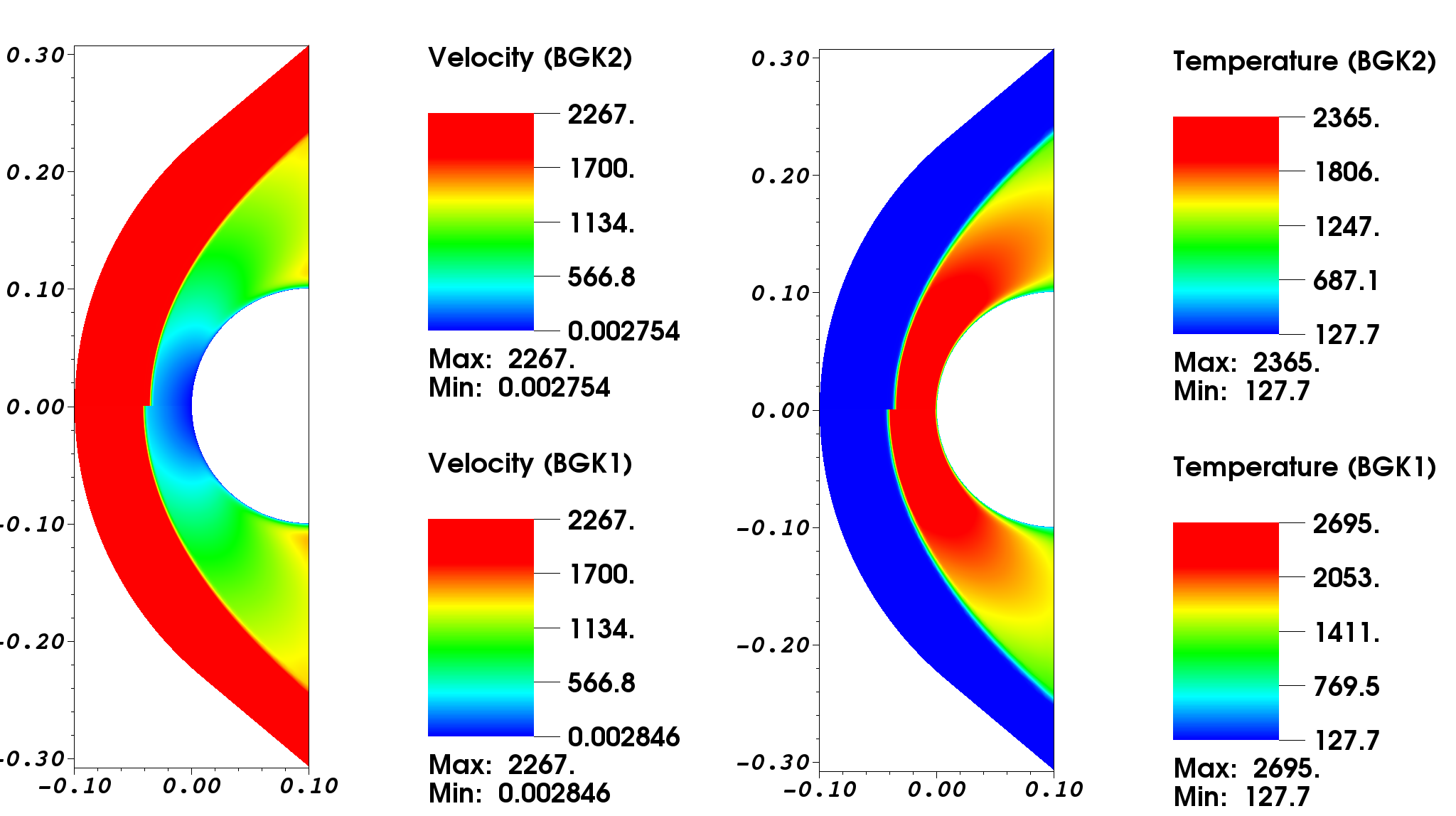}
\caption{Vibrating and non-vibrating air: velocity field and temperature field (Top: BGK2, bottom: BGK1)}
\label{bgkbgk2}
\end{figure}

\clearpage

\begin{figure}[p]
 \centering
 \includegraphics[width=0.9\textwidth]{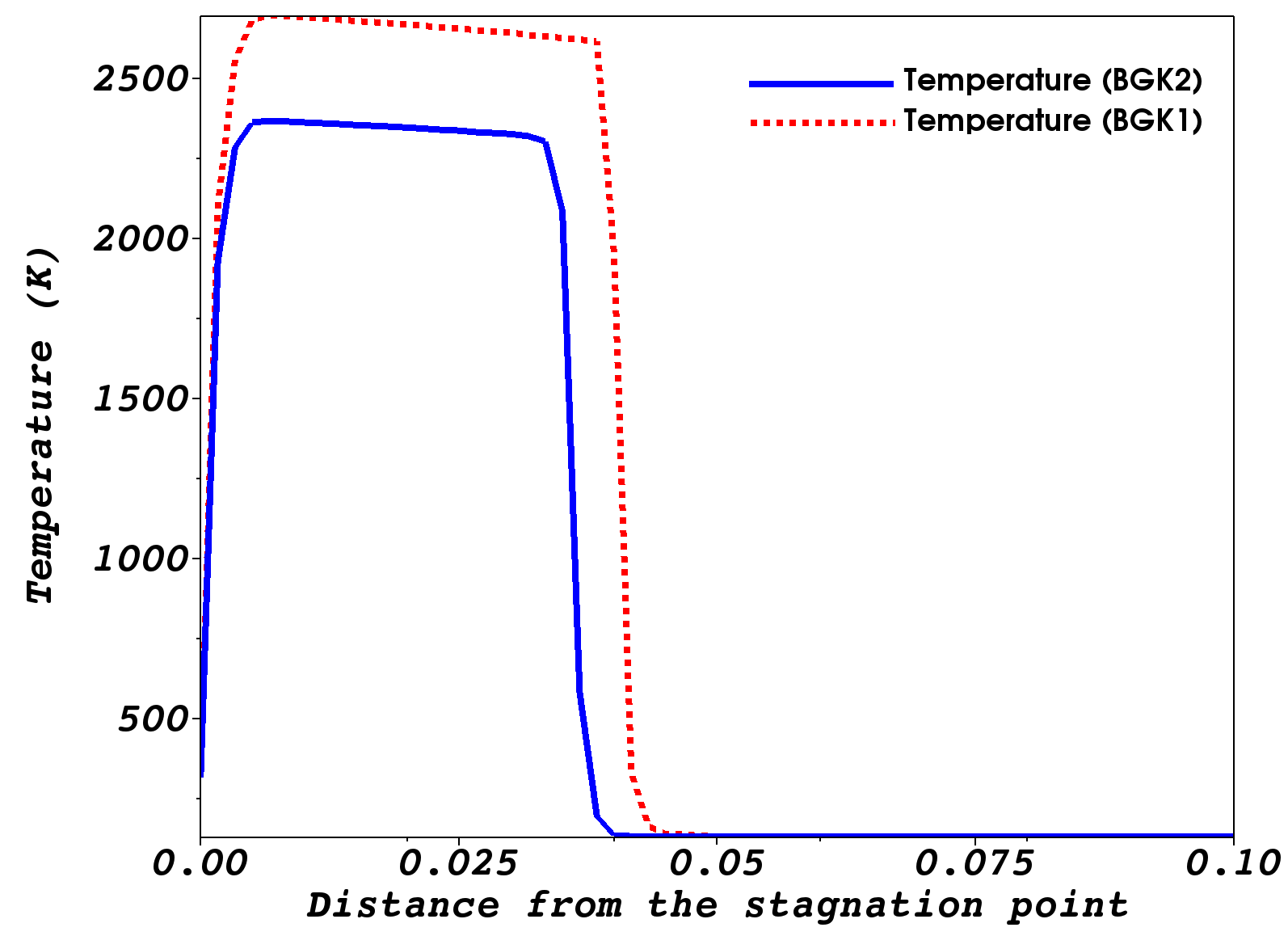}
 \caption{Vibrating and non-vibrating air: temperature profile along
   the stagnation line.}
\label{stagnation}
\end{figure}

\clearpage

\begin{figure}[p]
 \centering
 \includegraphics[width=0.5\textwidth]{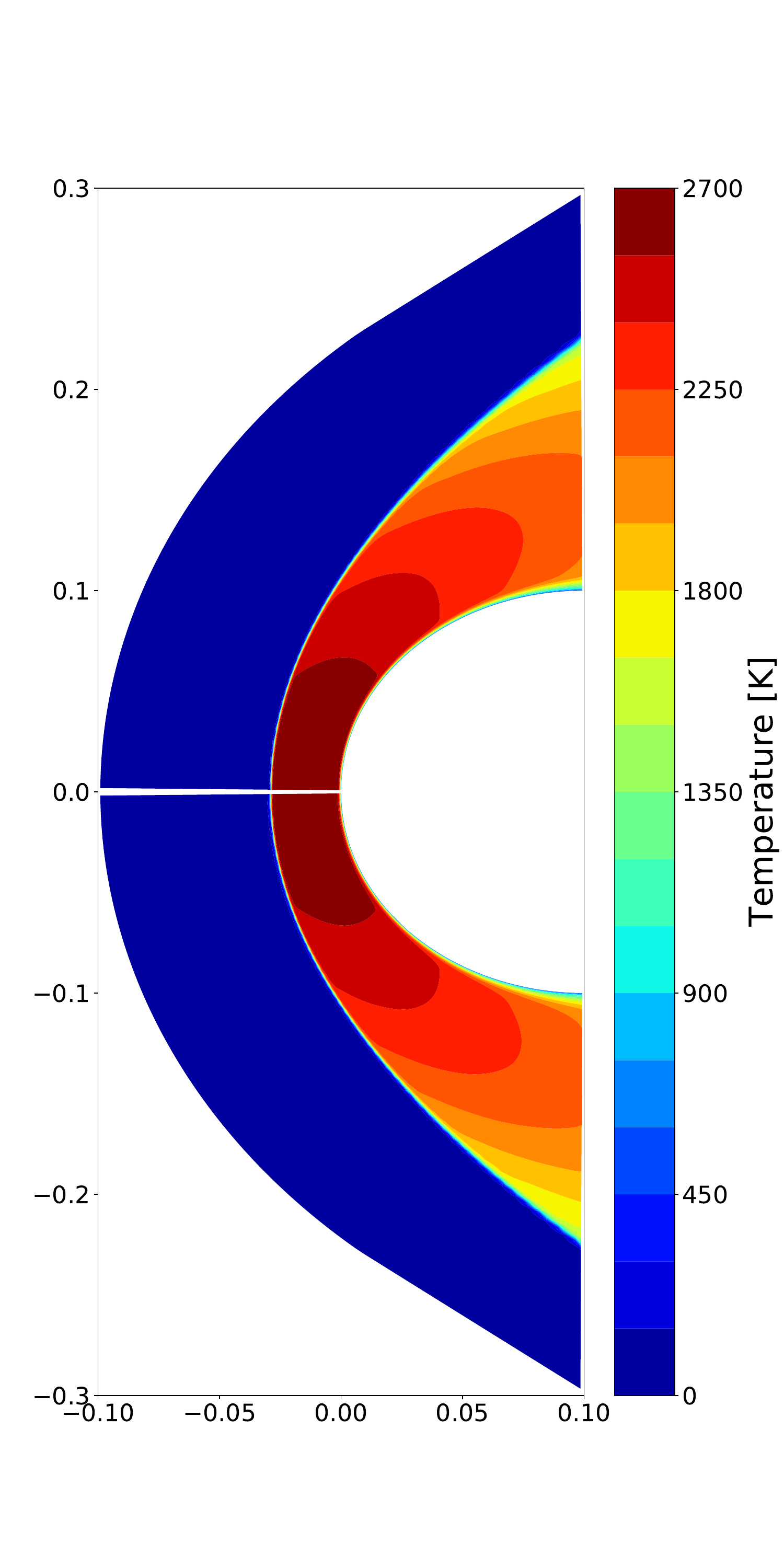}
 \caption{Dioxygen flow: temperature field obtained with a chemical
   equilibrium Navier-Stokes solver (top) and our chemical
   equilibrium BGK solver (bottom).}
\label{NSchimie_BGKchimie}
\end{figure}

\clearpage
\begin{figure}[p]
 \centering
 \includegraphics[width=0.9\textwidth]{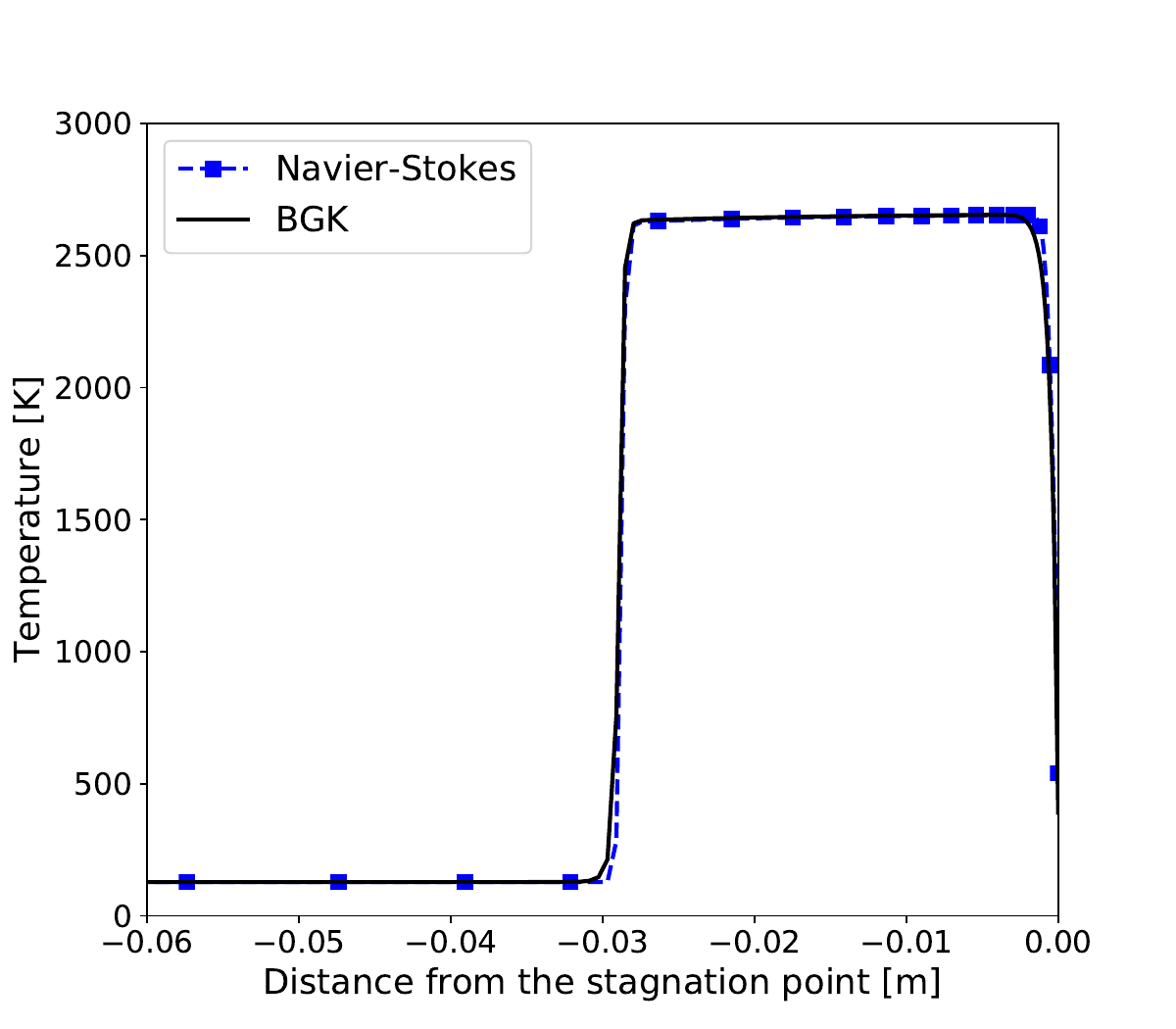}
 \caption{Dioxygen flow: temperature along the
   stagnation line.}
\label{T_profile}
\end{figure}

\clearpage
\begin{figure}[p]
 \centering
 \includegraphics[width=0.5\textwidth]{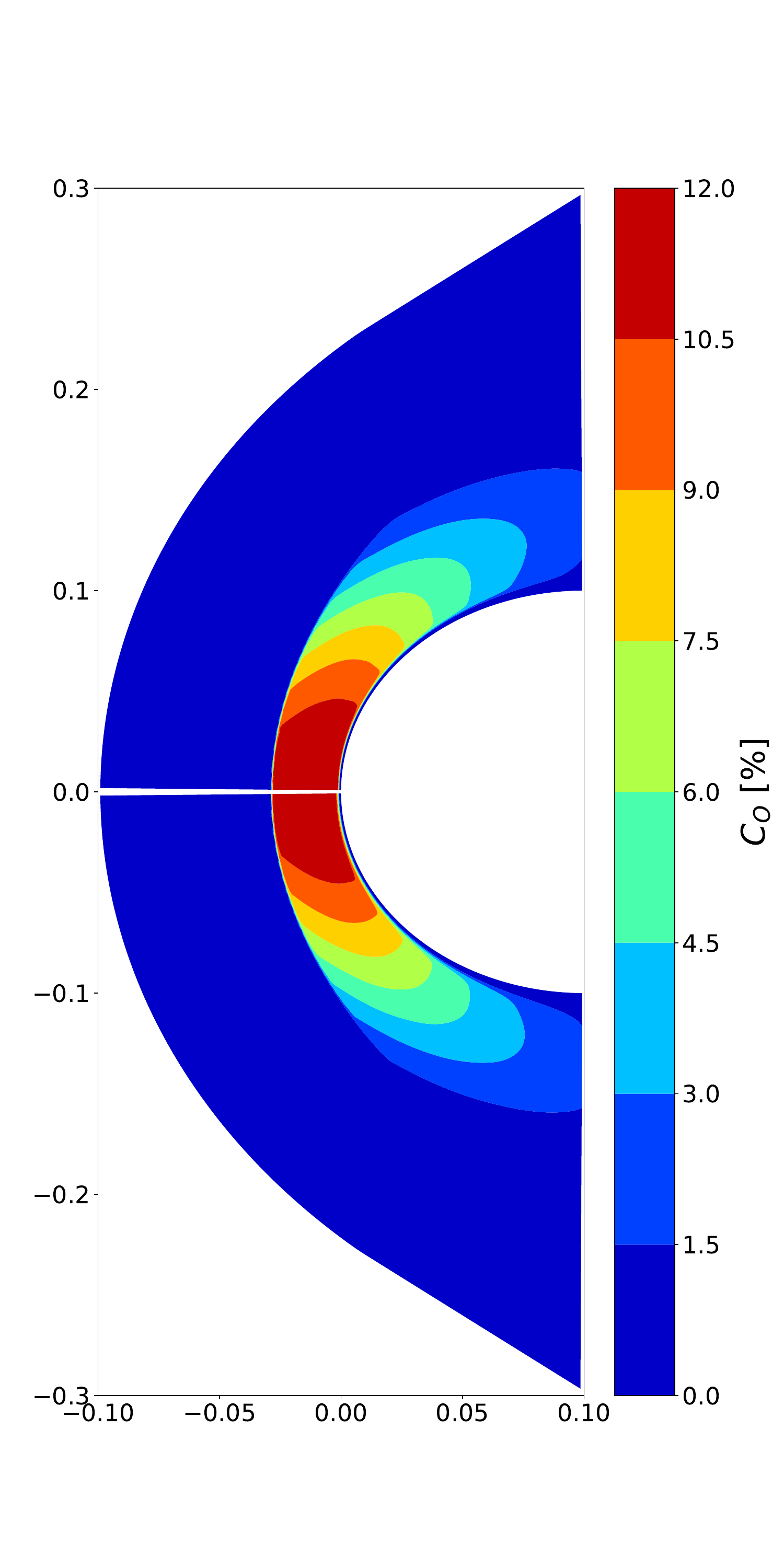}
\caption{Dioxygen flow: concentration of monoatomic oxygen obtained with a chemical
   equilibrium Navier-Stokes solver (top) and our chemical
   equilibrium BGK solver (bottom).}
\label{CO_NSchimie_BGKchimie}
\end{figure}

\clearpage
\begin{figure}[p]
 \centering
 \includegraphics[width=0.5\textwidth]{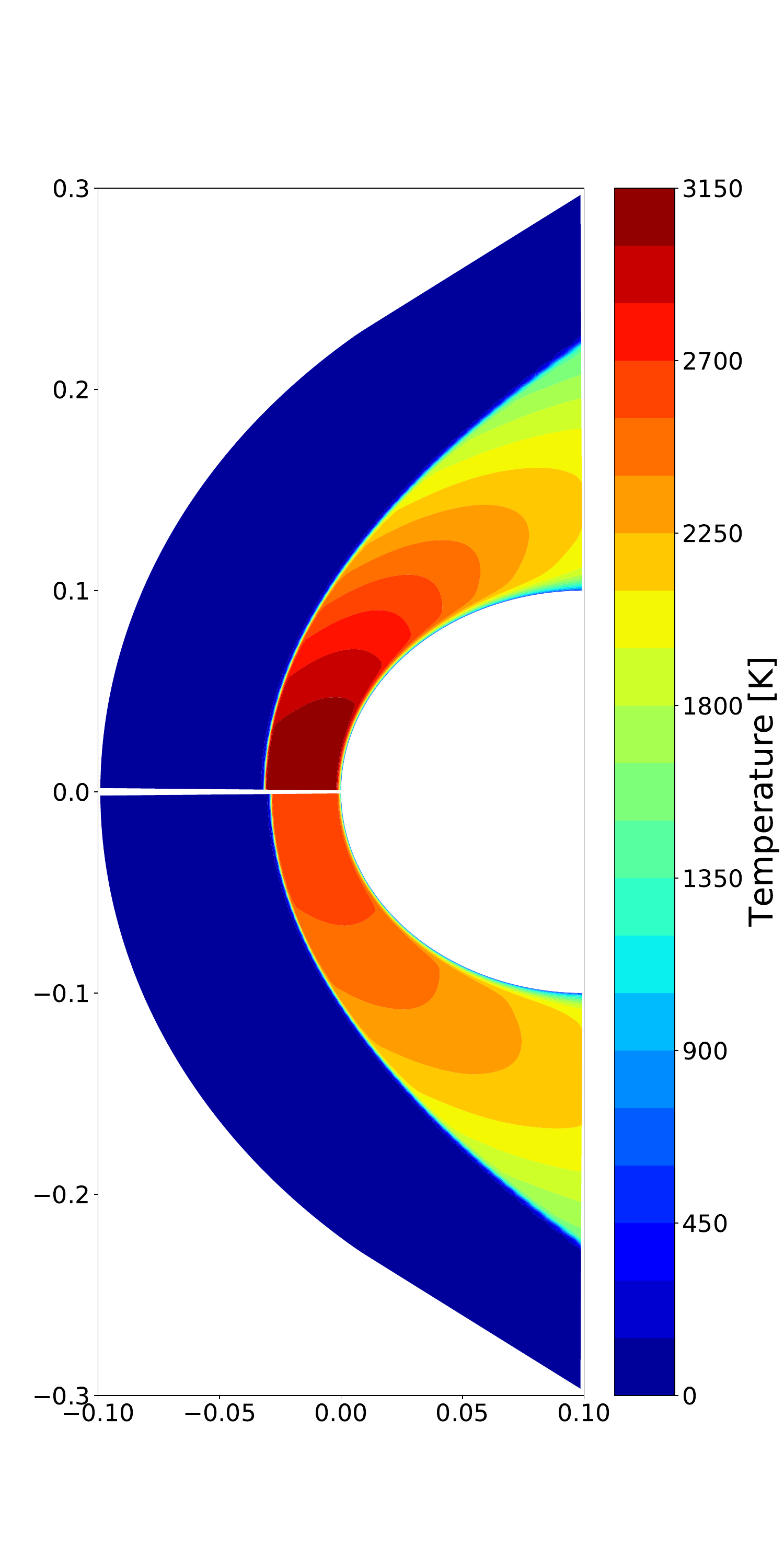}
 \caption{Dioxygen flow: temperature field obtained with our BGK
   solver with only
   vibrational energy (top), and our chemical
   equilibrium BGK solver (bottom).}
\label{BGKvibra_BGKchimie}
\end{figure}

\end{document}